\documentclass[journal,twoside,web]{ieeecolor}
\usepackage{lcsys}
\usepackage{cite}
\usepackage{amsmath,amssymb,amsfonts}
\usepackage{bm}
\usepackage{algorithm}
\usepackage{multirow}
\usepackage[noend]{algpseudocode}
\usepackage{graphicx}
\usepackage{subcaption}
\usepackage{textcomp}
\usepackage{subcaption}
\usepackage{comment}
\usepackage{colortbl}
\usepackage[dvipsnames]{xcolor}
\usepackage{hyperref}
\newtheorem{ass}{Assumption}

\newtheorem{lem}{Lemma}
\newtheorem{rem}{Remark}
\newtheorem{prop}{Proposition}
\newtheorem{cor}{Corollary}
\newtheorem{thm}{Theorem}
\newtheorem{Def}{Definition}

\definecolor{lightgray}{gray}{0.85}
\begin{document}
\title{\bf  Dynamically Feasible  Planning and Control in Complex Environments: a Scalable Systematic Approach}
\author{Miguel Castroviejo-Fernandez \IEEEmembership{Graduate Student Member, IEEE}, and Ilya Kolmanovsky \IEEEmembership{Fellow, IEEE}
\thanks{This work was supported by Air Force Office of Scientific Research, grant number FA9550-23-1-0678 and by the François-Xavier Bagnoud Doctoral Fellowship.}
\thanks{The authors are with the Department of Aerospace Engineering,
        University of Michigan, 48109 Ann Arbor MI, USA.
        {\tt\small mcastrov, ilya @umich.edu}}}
\maketitle
\thispagestyle{empty}

\begin{abstract}
In this article we present a method to generate safe sets for linear discrete-time systems subject to non-convex constraints that can be represented as a union of polytopes. It is then shown how a reference governor can be implemented for safe reference tracking tasks. A  theoretical analysis of the safe set is presented and properties of the reference governor scheme  are derived. The guarantees include safety at any time as well as finite-time convergence of the applied reference command to any strictly admissible reference command. For the proposed reference governor, online computational overhead is low. Moreover, it is shown that for specific instances of the complex constraint sets, the safe set can be computed efficiently. Extensive simulation results demonstrating the applicability of the method and online/offline computation times are reported.
\end{abstract}

\section{Introduction}
In the past decades, controllers have more and more frequently been expected to handle pointwise-in-time constraints. This has stimulated the widespread adoption of Model Predictive Control (MPC) \cite{mayne2000constrained} and Control Barrier Function schemes \cite{ames2019control}. Both methods produce control inputs that ensure constraint satisfaction. 
Another approach to enforce constraint satisfaction are reference governors (RG), see \cite{garone2017reference} and references therein. RGs are supervisory schemes that act as reference command filters modifying, when necessary, the reference command applied to the system in order to ensure constraint satisfaction. The reference command selection strategy in most RG schemes is tightly interconnected with the notion of Maximum Output Admissible Sets (MOAS) \cite{gilbert1991linear}. The MOAS is defined as the set of all initial condition and reference command pairs such that the resulting trajectory when the reference command is kept constant does not result in any constraint violation. 

While the above constrained control methods are well established and are backed by strong theoretical guarantees in the case of convex (in particular polytopic) constraint sets the handling of non-convex constraint sets is still an active field of research. Due to the inherent  complexity of the non-convex case it is common in the planning community to discard the dynamics and focus on feasible path generation \cite{weiss2015safe,danielson2020motion,marcucci2024fast}.  
Moreover, most studies \cite{blanchini2004control,wu2025towards,wang2022geometrically,marcucci2024fast} are restricted to geometrical constraints, i.e. constraints on the position of a system while ignoring other dynamical constraints, this restriction is not necessary with our approach.

In this work we develop a general framework for generating safe sets for linear time invariant systems subject to pointwise-in-time constraints. The constraint set, while non-convex, is assumed to be represented as a union of polytopes with each pair of polytopes in the union overlapping over at most one facet. Under suitable connectedness properties between the elements of the union, the constraint set is dubbed a \textit{connected collection of polytopes}. 
The safe set constructed for this collection utilizes the MOAS of simpler sets generated from the polytopes in the collection. Given its root in the MOAS, the computed safe set is usually large. The safe set is computed offline and can then be used online for safe reference tracking tasks.
The  contributions of this work are as follows.\begin{itemize}
    \item We establish a systematic and general approach to generate a safe set for any connected collection of polytopes. The method relies on defining sets linking adjacent polytopes in the collection that we refer to as \textit{(weak) extensions}. The (weak) extensions provide a safety path between the different polytopes. We derive 
    properties of the (weak) extensions and, in particular, show that they are themselves polytopes with non-empty interior. 
    \item  Following the philosophy of RGs we develop a supervisory scheme that enables safe reference tracking capability, leveraging the computed safe set. Moreover, under suitable assumptions we prove that starting from a safe state and reference command pair the applied reference command converges, in finite time, to any strictly admissible constant reference command. 
\item In the case where the connected collection of polytopes is solely composed of hyperrectangles we show that the safe set computations reduce to scaling and centering a small number of base polytopes. This makes the approach amenable to online implementation.
\item We provide extensive numerical simulations for several dynamic systems including on-orbit proximity maneuvering and quadcopter control in an urban environment.\footnote{The code for the simulations in Section \ref{sec:CWH} and Section \ref{sec:droneInCity} can be found at \url{https://github.com/mcastrov-pixel/Reference-Governor-for-union-of-polytopes}}
\end{itemize} 
It must be noted that approximating non-convex regions by a finite union of polytopes is not a new idea, for example, Blanchini \textit{et al.} \cite{blanchini2004control} propose a control scheme based on this for a class of robot manipulators systems. 
Moreover, RG schemes that handle non-convex constraint sets have also been proposed in the past.  In particular, RG schemes for intersections and unions of concave constraints were studied in \cite{hosseinzadeh2019constrained,hosseinzadeh2019explicit}.
A scalar reference governor strategy for linear systems subject to constraints represented by a union of polytopes has been considered in \cite{romagnoli2020new} under the assumption that ``touching" polytopes in the sequence overlap with nonempty interior. In contrast with our approach, this assumption in \cite{romagnoli2020new} avoids the challenge of generating a safe transition regions.
The authors of \cite{nguyen2021b} introduce a similar notion of extension for a strictly connected sequence of polytopes and use it as a tool for path generation. Nevertheless, their main focus is on smooth path generation with static constraint satisfaction while the properties of the extension are not studied in \cite{nguyen2021b}. Relative to \cite{nguyen2021b}, we make the following novel contributions: We introduce the weak extension which enables us to handle polytopes with partial overlap, we derive theoretical properties of the (weak) extensions, and, moreover, through the MOAS we take into account the dynamics of the system and are able to handle non-geometric constraints directly.

The article is structured as follows: We conclude this section by introducing notation as well as some basic facts about polytopes. Section \ref{sec:DTLTI_and_RG} recalls how a classical RG scheme for a discrete-time linear time-invariant (DT LTI) system is constructed. Section \ref{sec:connectedSeqDef} introduces the class of constraint sets we consider and other notions instrumental in the safe set generation. The supervisory scheme that we propose is described in Section \ref{sec:RGNonConvex}. We then derive theoretical properties of the connected collection of polytopes (Section \ref{sec:propsSequence}) and of the proposed RG (Section \ref{sec:propsRG}). Section \ref{sec:hyperrectangles} provides a method for fast safe set generation. Simulation results are reported in Section \ref{sec:numericalSim}.
\subsection{Notation}
Let $\mathbb Z$ be the set of integers and $\mathbb R$ be the set of real numbers. Given sets $\mathbb S,\mathcal A\subseteq \mathbb R^n$, then $\mathbb S_{\mathcal A}\triangleq \mathbb S\cap \mathcal A$ and for $a\in\mathbb R$ and $\mathbb S\subseteq\mathbb R$, $\mathbb S_{\geq a} \triangleq \mathbb S_{[a,\infty)}$ and $\mathbb S_{> a} \triangleq \mathbb S_{(a,\infty)}$.  Given a set $\mathcal S$ the power set of $\mathcal S$ is denoted $\mathcal P(\mathcal S)$ and the relative interior of $\mathcal S$ is denoted ${\tt ri}(\mathcal S)$.
For a given matrix $A\in\mathbb R^{m\times n}$ its $i^{\text{th}}$ row is denoted by $A_{(i)}$. Moreover, $A_{(-i)}\in\mathbb R^{m-1\times n}$ is the matrix obtained by removing the $i^{\text{th}}$ row from $A$. The hyperplane generated by the row vector $0\neq a\in\mathbb R^{ n}$ and scalar $b\in\mathbb R$ is denoted as ${\tt H}(a,b) \triangleq \{y\in\mathbb R^{n}:\, a y = b\}$. Given a matrix $A\in\mathbb R^{m\times n}$ and a vector $b\in \mathbb R^m$, the polytope generated by the pair $A,b$ is ${\tt P}(A,b) = \{x\in\mathbb R^n:\, Ax\leq b\}$.
Matrix and scalar multiplications of $\mathcal A\subseteq \mathbb R^n$ by $B\in\mathbb R^{m\times n}$ and $\beta\in\mathbb R$ are defined as $B\mathcal A = \{Bx:\;x\in\mathcal A\}$ and  $\beta \mathcal A = \{\beta x:x\in\mathcal A\}$, respectively. The identity matrix for $\mathbb R^{n\times n}$ is denoted $I_n$ and $0_{n\times m}$ denotes an $n\times m$ matrix with zero in each entry; when clear from the context the subscripts are omitted. Given $x\in \mathbb R^n$, The 2-norm is defined as $\|x\| = \sqrt{x^\top x}$, the 2-norm unit ball centered at $x$ is $\mathcal B(x)\triangleq \{y\in\mathbb R^n:\,\|(x-y)\|_2\leq 1\}$ and $\mathcal B = \mathcal B(0)$. The Minkowski sum of two sets  $A,\;B\subset\mathbb R^n$ is defined as $A\oplus B=\{a+b:\;a\in A,\;b\in B\}.$
\subsection{Polytopes}
A set $\mathcal A\subseteq\mathbb R^{n}$ is a polytope if it can be expressed as the intersection of a finite number of  half-spaces. The dimension of a polytope corresponds to the dimension of its affine hull. It is full-dimensional if it has dimension $n$.
An  H-representation of $\mathcal A$ is a pair $(A^{\mathcal A},b^{\mathcal A})\in \mathbb R^{m\times n}\times \mathbb R^{m}$ such that $\mathcal A = {\tt P}(A^{\mathcal A}, b^{\mathcal A})$.
It is a minimal representation of $\mathcal A$ if no H-representations can be constructed with a strictly smaller number of  half-spaces. For a full-dimensional polytope, the minimal H-representation is unique up to scalar multiplication and re-ordering. In the following, $(A^{\mathcal A},b^{\mathcal A})$ denotes a minimal H-representation of $\mathcal A$. 
A $k$-face of $\mathcal A$, $k\in\mathbb Z_{[0,n-1]}$, is any non-empty set that can be expressed as $\bigcap_{j\in\mathcal J}{\tt H}(A^{\mathcal A}_{(j)},b^{\mathcal A}_{(j)})\cap\mathcal A$ where the index set $\mathcal J$ has cardinality $n-k$. Furthermore, the $0$-faces of $\mathcal A$ are vertices, the $1$-faces are called edges and the $n-1$-faces are called facets. If $\mathcal A$ is bounded it is described by its V-representation: as the convex hull of the set of its vertices.

\section{Reference governor for systems with non convex constraints}
In this section we describe the proposed supervisory scheme for reference tracking in linear systems subject to non-convex constraint sets. Section \ref{sec:DTLTI_and_RG} summarizes the ideas underlying reference governors for DT LTI systems subject to polytopic constraints. In Section \ref{sec:connectedSeqDef} we precisely define the class of non-convex constraint sets we consider, and introduce the proposed scheme in Section \ref{sec:RGNonConvex}.
\subsection{Reference governor for discrete-time linear systems}
\label{sec:DTLTI_and_RG}
Consider a DT LTI system with dynamics: 
\begin{subequations}\label{eq:dyns}
\begin{align}
    x_{k+1} &= A x_k + B v_k,\\
    y_k & = C x_k + D v_k,
\end{align}
\end{subequations}
where $k\in \mathbb Z_{\geq 0}$ denotes the time instant, $x_k\in\mathbb R^{n_x}$ is the state of the system, $v_k\in\mathbb R^{n_v}$ is a control input to the system and $y_k\in\mathbb R^{n_y}$ is the output vector. The output is subject to pointwise-in-time  constraints described by the set $\mathcal Y$, i.e.,
\begin{align}\label{eq:cstr}
y_k\in\mathcal Y, \text{for all }k\in\mathbb Z_{\geq 0}.    
\end{align}
For the time being we make the following assumption on $\mathcal Y$.
\begin{ass}
    \label{ass:simpleCstr} The set $\mathcal Y\subset\mathbb R^{n_y}$ is a bounded polytope with H-representation given by $(A^{\mathcal Y},b^{\mathcal Y})$. Moreover, $\mathcal Y$ contains the origin in its interior, i.e., $b^{\mathcal Y}\in \mathbb R^{n_y}_{>0}$.
\end{ass}

Note that, in general, compactness of the constraint set is a reasonable requirement as one can include extra box constraints to restrict the unbounded directions. We make the following assumption on \eqref{eq:dyns}.
\begin{ass}\label{ass:schur}
    $(C,A)$ is observable and $A$ is Schur.
\end{ass}
This is motivated by our focus on the task of reference tracking for a pre-stabilized system. In this setting, it is usual to let the input represent the applied reference command at time $k$. Then, we introduce the set of (strictly) admissible reference commands  parametrized by the margin $\epsilon\in\mathbb R_{\geq 0}$:
\begin{align}
        \mathcal R(\mathcal Y,\epsilon)  &= 
        \{v\in\mathbb R^{n_v}:\, \{H_\infty v\}\oplus\epsilon\mathcal B\subseteq\mathcal Y \}, \text{where}\nonumber \\
       H_\infty &=  D+ C(I-A)^{-1}B,\label{eq:Hinfty}
\end{align}
and Assumption \ref{ass:schur} ensures that $(I-A)$ is invertible. The set of feasible reference commands is $\mathcal R(\mathcal Y,0)$. When clear from context we may omit the second argument of $\mathcal R(\mathcal Y,\epsilon)$.

Reference governor schemes \cite{garone2017reference} select a reference command that is safe, given the current state of the system. More precisely, we say that $(v,x)\in \mathbb R^{n_v}\times\mathbb R^{n_x}$ is safe for the dynamics \eqref{eq:dyns} and constraints \eqref{eq:cstr} if the output trajectory starting from $x$ and with constant input $v$ does not lead to constraint violation. The set of all safe reference command and state pairs is the Maximum Output Admissible Set (MOAS) \cite{gilbert1991linear}:
     \begin{align*}
    \mathcal O(\mathcal Y)&\!\triangleq \!\{(v,x)\in\mathbb R^{n_v+n_x}:\, C A^k x\! + \!H_k v \!\in\!\mathcal Y\;\forall k\in\mathbb Z_{\geq 0}\},\\
    H_k &= C(I_n-A)^{-1}(I-A^k)B +D. \label{eq:Hk}
\end{align*}
While, in general, $\mathcal O(\mathcal Y)$ may not be computable in finite time, an $\epsilon-$close approximation of the MOAS can be \cite{gilbert1991linear}. The approximation is obtained by a slight tightening of the set of feasible reference commands, i.e., 
\begin{equation}
    \mathcal O(\mathcal Y,\epsilon) = \mathcal O(\mathcal Y)\cap  (\mathcal R(\mathcal Y,\epsilon)\times \mathbb R^{n_x}).
\end{equation}
For the reference tracking task, let $r\in\mathbb R^{n_v}$ be the desired setpoint.  
As mentioned above, RG schemes select the applied reference command, $v_k$, in such a way that if the reference command is maintained safety is guaranteed. One example of RG scheme is the command governor (CG), which, given the current state vector $x_k$, chooses the applied reference command as:
\begin{align*}
    v_k &= \text{arg}\min_{v}\|v-r\|^2 \text{ subject to } (v,x_k)\in\mathcal O(\mathcal Y,\epsilon).
\end{align*}
For the CG, \cite[Proposition 1]{bemporad1997nonlinear} is a classical result providing conditions under which strictly steady-state admissible reference commands can be reached in finite time. Hereunder, we give a weaker version of the result that holds directly from     \cite[Proposition 1]{bemporad1997nonlinear}.
\begin{thm}\label{thm:convToStrictSteadyState}
Let Assumptions \ref{ass:simpleCstr} and \ref{ass:schur} hold and let $\epsilon\in\mathbb R_{>0}$ be given. Consider the closed loop trajectory obtained from $x_0\in\mathbb R^{n_x}$ through \eqref{eq:dyns} and with $v_k$ generated from the CG for a constant $r\in\mathcal R(\mathcal Y,\epsilon)$. If there exists $v^0$ such that $(v^0,x_0)\in\mathcal O(\mathcal Y)$, then, $v_k$ converges in finite time to $r$. More precisely, 
there exists $ k^\star\in\mathbb Z_{\geq 0}$ such that $v_k = r$ for all $k\geq k^\star$.
\end{thm}
Similar results can be derived for other RG schemes such as the scalar reference governor, explicit reference governor \cite{garone2017reference} and other reference governor schemes for nonlinear systems \cite{castroviejo2024safe,castroviejo2025robust}. Nevertheless, this usually assumes a convex constraint set. 
In this work, we aim to derive similar results for more complex constraint sets representable as a finite union of convex sets, i.e., 
\begin{equation*}
    \mathcal Y = \bigcup_{i=1}^{n_{\mathcal Y}}\mathcal Y_i,\quad \mathcal Y_i\subseteq \mathbb R^{n_y},\;i=1,\dots,\;n_{\mathcal Y}.
\end{equation*}
The set $\mathcal Y$ is associated with the following collection of sets,
\begin{equation*}
   \{\mathcal Y_i\}_{i=1}^{n_{\mathcal Y}}\triangleq \{\mathcal Y_i\subseteq \mathbb R^{n_y}: i=1,\dots,n_{\mathcal Y}\}.
\end{equation*}
In order to derive constructive results we restrict the class of collections, $\{\mathcal Y_i\}_{i=1}^{n_{\mathcal Y}}$, to a connected collection of polytopes.
\begin{figure*}
\centering
\begin{subfigure}{.45\textwidth}
  \centering
 \includegraphics[width=1\linewidth]{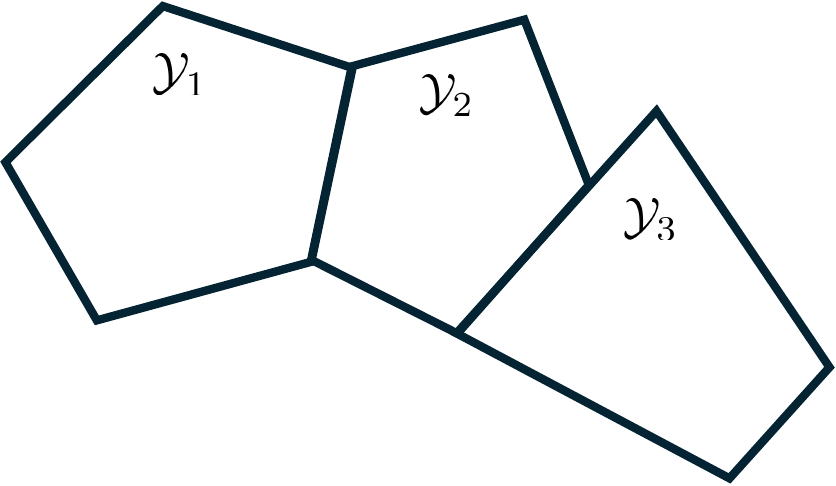}
    \caption{{\small connected collection of polytopes with $n_{\mathcal Y}=3$.}}
    \label{fig:connectedSequence}
\end{subfigure}
$\;$
\begin{subfigure}{.45\textwidth}
  \centering
  \includegraphics[width=1\linewidth]{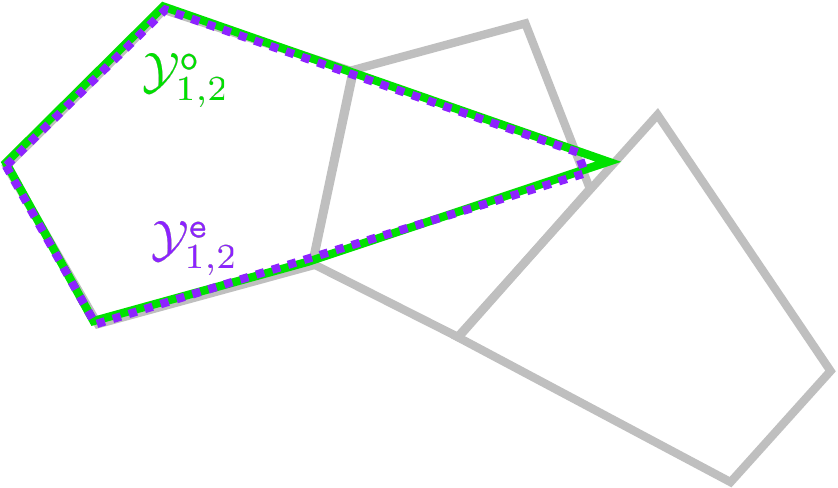}
    \caption{\small Openings and extensions for the collection of polytopes}
    \label{fig:Ext1}
\end{subfigure}
$\;$
\begin{subfigure}{.45\textwidth}
  \centering
  \includegraphics[width=1\linewidth]{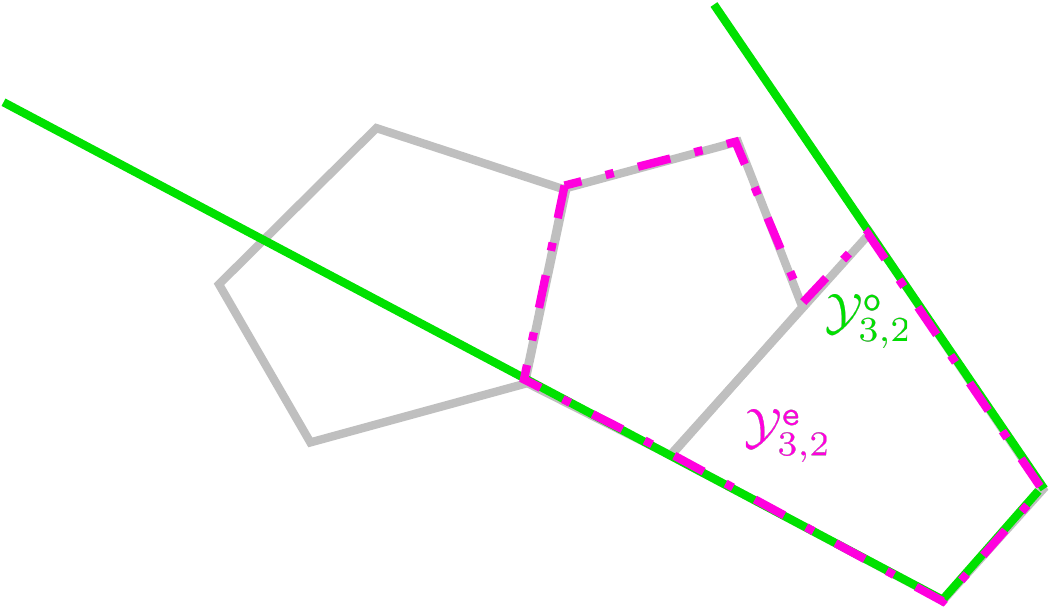}
  \caption{Opening and extension for the collection of polytopes }
  \label{fig:Ext2}
\end{subfigure}%
\caption{Illustration of a connected collection of polytopes with $\{\mathcal Y_1,\;\mathcal Y_2\}$ strictly connected but $\{\mathcal Y_2,\;\mathcal Y_3\}$ not strictly connected}
\label{fig:exampleConnectedSeq}
\end{figure*}
\subsection{Connected collection of polytopes}
\label{sec:connectedSeqDef}
We start this section by introducing some definitions relating to a collection of sets $\{\mathcal Y_i\}_{i=1}^{n_\mathcal Y}$.
\begin{Def}
    \label{def:adjacencyMap} Let $\{\mathcal Y_i\}_{i=1}^{n_{\mathcal Y}}\subset\mathcal P(\mathbb R^{n_y})$, where $n_{\mathcal Y}\in\mathbb Z_{>0}$. The \textbf{adjacency matrix} of $\{\mathcal Y_i\}_{i=1}^{n_{\mathcal Y}}$, denoted by $F(\{\mathcal Y_i\}_{i=1}^{n_{\mathcal Y}})\in \ \{0,1\}^{n_{\mathcal Y}\times n_{\mathcal Y}}$, is a binary matrix with entries given by 
    \begin{equation}\label{eq:FijDef}
        F_{i,j} = \begin{cases}
            &1,\text{ if } \mathcal Y_i\cap \mathcal Y_j\neq \emptyset,\\
            &0, \text{ else}
        \end{cases}.
    \end{equation}
\end{Def}
The adjacency matrix describes which elements of $\{\mathcal Y_i\}_{i=1}^{n_\mathcal Y}$ have non-empty intersection. When clear from the context we omit its argument.
\begin{Def}\label{def:connectedSec}
        Let $\{\mathcal Y_i\}_{i=1}^{n_{\mathcal Y}}\subset\mathcal P(\mathbb R^{n_y})$ where $n_{\mathcal Y}\in\mathbb Z_{>0}$.  We say that $\{\mathcal Y_i\}_{i=1}^{n_{\mathcal Y}}$ is a \textbf{(simply) connected collection of polytopes} if \begin{itemize}
        \item[\ref{def:connectedSec}.a] For $i=1,\dots,n_{\mathcal Y}$, $\mathcal Y_{i}$ has non-empty interior, i.e., it is a full-dimensional polytope.
        \item[\ref{def:connectedSec}.b] The adjacency matrix represents a connected graph.
        \item[\ref{def:connectedSec}.c] For every pair $(i,j)$ such that $F_{i,j}= 1$ and $i\neq j$, $\mathcal Y_i$ and $\mathcal Y_{j}$ intersect on a single facet and the intersection has dimension $n_y-1$. More precisely, there is $k\in\mathbb Z_{[1,n_{\mathcal Y_i}]}$, $\epsilon>0$ and $y\in{\tt H}(A^{\mathcal Y_i}_{(k)},b^{\mathcal Y_i}_{(k)})$ such that 
        \begin{equation}
        (\{y\}\oplus\epsilon\mathcal B)\cap{\tt H}(A^{\mathcal Y_i}_{(k)},b^{\mathcal Y_i}_{(k)})\subseteq\mathcal Y_i\cap\mathcal Y_{j}\subseteq {\tt H}(A^{\mathcal Y_i}_{(k)}  ,b^{\mathcal Y_i}_{(k)}), \label{eq:overlapCondConnectedSeq}
        \end{equation}
        where $n_{\mathcal Y_i}$ is the number of rows of $A^{\mathcal Y_i}$.   
    \end{itemize}
\end{Def}
\begin{Def}\label{def:strictConnectedSec}
    $\{\mathcal Y_i\}_{i=1}^{n_{\mathcal Y}}$ is a \textbf{strictly connected collection of polytopes} if it is a connected collection of polytopes, and, for every $i,j$ such that $F_{i,j}= 1$ the polytopes $\mathcal Y_i,\;\mathcal Y_{j}$ have a \textit{common facet}, i.e., there is $k\in \mathbb Z_{[1,n_{\mathcal Y_i}]}$ such that
    \begin{equation*}
        {\tt H}(A^{\mathcal Y_i}_{(k)},b^{\mathcal Y_i}_{(k)})\cap\mathcal Y_i = {\tt H}(A^{\mathcal Y_{i}}_{(k)},b^{\mathcal Y_{i}}_{(k)})\cap\mathcal Y_{j}.
    \end{equation*}
\end{Def}

The intersection between adjacent elements of a connected collection of polytopes will play a significant role in our derivations. For convenience we refer to it as the \textbf{gate} between the elements. We  also introduce the \textbf{connection matrix}, $G(\{\mathcal Y_i\}_{i=1}^{n_{\mathcal Y}})$, representing the indexes of the common hyperplane between two elements; when clear from the context we ignore its argument. The entries of the connection matrix are given by 
\begin{align*}
G_{i,j} = \begin{cases}
    &1, \;\text{ if }F_{i,j} = 0 \text{ or } i =j\\
    & k \text{ such that \eqref{eq:overlapCondConnectedSeq} holds for } \mathcal Y_i,\; \mathcal Y_j, \text{ else} 
\end{cases}.
\end{align*}
Then, for any $i\neq j$ such that $F_{i,j}=1$, up to a positive constant:
\begin{equation*}
    A^{\mathcal Y_i}_{G_{i,j}} =- A^{\mathcal Y_{j}}_{G_{j,i}},\quad b^{\mathcal Y_i}_{G_{i,j}} =  b^{\mathcal Y_{j}}_{G_{j,i}}.
\end{equation*}
Figure \ref{fig:connectedSequence} depicts a connected collection of polytopes in $\mathbb R^2$ where $\{\mathcal Y_1,\mathcal Y_2\}$ are strictly connected and $\{\mathcal Y_2,\mathcal Y_3\}$ are connected. Finally, we introduce some sets that play a central role in generating safe sets. 
\begin{Def}\label{def:auxSets}
    Let $\{\mathcal Y_i\}_{i=1}^{n_{\mathcal Y}}$ be a connected collection of polytopes and let $i,j\in\mathbb R_{[1,n_{\mathcal Y}]}$ be such that $F_{i,j} = 1$ and $i\neq j$.  
    The \textbf{opening} of $\mathcal Y_i$ towards $\mathcal Y_j$ is defined as \begin{equation*}
        \mathcal Y_{i,j}^{\tt o}\triangleq \{y \in\mathbb R^{n_y}:\,A^{\mathcal Y_i}_{(-G_{i,j})}y\leq b^{\mathcal Y_i}_{(-G_{i,j})}\},
    \end{equation*}
    and the \textbf{restriction} of $\mathcal Y_j$ by $\mathcal Y_i$ is defined as
    \begin{equation*}
   \mathcal Y_{j}^{\tt res} =   \mathcal Y^{\tt o}_{i,j}\cap \mathcal Y_{j}.    
    \end{equation*}
    
    Furthermore, the \textbf{ extension} of  $\mathcal Y_i$ towards $\mathcal Y_j$ is defined as 
    \begin{equation*}
        \mathcal Y^{\tt e}_{i,j}\triangleq \mathcal Y_i\cup \mathcal Y^{\tt res}_{j,i},
    \end{equation*} and the \textbf{weak extension} of the pair $(\mathcal Y_i,\mathcal Y_{j})$ is
    \begin{equation*}
        {\tt WE}(\mathcal Y_i,\mathcal Y_{j})\triangleq \mathcal Y^{\tt res}_{i,j}\cup \mathcal Y^{\tt res}_{j,i}.
    \end{equation*}
\end{Def}

Loosely speaking, the opening of $\mathcal Y_i$ towards $\mathcal Y_j$ is the polytope obtained from $\mathcal Y_i$ when we remove the hyperplane that intersects with $\mathcal Y_{j}$. The restriction of $\mathcal Y_j$, is what remains of $\mathcal Y_j$ after we intersect it with the opening $\mathcal Y^{\tt o}_{i,j}$. The extension of $\mathcal Y_i$ towards $\mathcal Y_j$ restricts $\mathcal Y^{\tt o}_{i,j}$ to the points that are either in $\mathcal Y_i$ or in $\mathcal Y_{j}$. Figures \ref{fig:Ext1}-\ref{fig:Ext2} illustrate this. In Figure \ref{fig:exampleConnectedSeq}, we note that all openings are convex. Moreover, the extensions relating to the strictly connected subcollection, that is $\mathcal Y^{\tt e}_{1,2},\;\mathcal Y^{\tt e}_{2,1}$, are also convex and so is $\mathcal Y^{\tt e}_{2,3}$. In contrast, $\mathcal Y^{\tt e}_{3,2}$ is not convex. This motivates the introduction of the weak extension. As we will see later, the weak extension corresponds to the extension for strictly connected subsets of the original collection. An illustration of the weak extension is shown in Figure \ref{fig:exampleWeakExtension}.  

Properties of the (weak) extensions are formally derived in Section \ref{sec:theoreticalProps}. However, we mention that for a strictly connected collection, the extensions are indeed polytopes. For a simply connected collection these results do not apply to the extension but hold for the weak extension. Crucially, the (weak) extension provides a natural way to construct a bridge between neighboring sets. Given that the \textit{overlap} has a non-empty interior (see Proposition \ref{prop:overlapIsNonEmpty}) we will be able to use the (weak) extension to ensure a safe transition from one element to the next.  
In other words, the safe set for $\mathcal Y$ will be constructed using the individual MOAS for the sets $\mathcal Y^{\tt e}_{i,j}$ or ${\tt WE}(\mathcal Y_{i},\mathcal Y_{j})$. For the approach to be successful, we need to make sure the equilibrium manifold adequately intersects the collection. The following definition formalizes this idea.
\begin{Def}\label{def:steadyOutputCompliant}
    The connected collection of polytopes $\{\mathcal Y_i\}_{i=1}^{n_{\mathcal Y}}$ is \textbf{steady-output compliant} for the dynamics \eqref{eq:dyns} if the steady-output manifold intersects the relative interior of each element and gate of the connected collection, i.e., 
    \begin{align}\label{eq:gateCompliance}
        &\emptyset\neq H_{\infty} \mathbb R^{n_v}\cap \mathcal {\tt ri}(\mathcal Y_i),\; i = 1,\dots,n_{\mathcal Y}\,\nonumber\\
        &\emptyset\neq H_{\infty} \mathbb R^{n_v}\cap  {\tt ri}(\mathcal Y_i\cap\mathcal Y_j) \forall i\neq j\in\mathbb Z_{[1,n_{\mathcal Y}]} \text{ s.t. }F_{i,j}=1, \nonumber
    \end{align}
    where  $H_{\infty}$ is defined in \eqref{eq:Hinfty}.
\end{Def}
We then introduce the following assumptions.
\begin{ass}\label{ass:connectedSequence}
    The constraint set $\mathcal Y$, in \eqref{eq:cstr}, is represented by a connected collection of polytopes, $\{\mathcal Y_{i}\}_{i = 1}^{n_{\mathcal Y}}\subseteq \mathbb R^{n_y}$. 
    Moreover, $\mathcal Y$ is steady-output compliant for the dynamics \eqref{eq:dyns}.
     \end{ass}
     
 There exists a wide class of practical applications such that Assumption \ref{ass:connectedSequence} holds, as exemplified in Section \ref{sec:CWH} for an aerospace system and in Section \ref{sec:droneInCity} for a robotic system.
\subsection{Reference governor for non-convex feasible region}
\label{sec:RGNonConvex}
We now introduce a reference governor scheme for constraint enforcement when the feasible region is represented by a connected collection of polytopes. 
The overarching idea is to advance between different elements of the connected collection, leveraging the connectedness and the weak extension between elements. The weak extensions enable safe transition from one element to the next. 
The connectedness of the collection ensures a path between the starting point and desired setpoint can be computed. A sequence of intermediate reference commands is then generated with two elements inside each weak extension along the path, one on each side of the gate. A reference governor, such as that described in Section \ref{sec:DTLTI_and_RG}, with  constraint set formed from either an element of the collection or a weak extension is then used to navigate to the next intermediate reference command.

In the sequel, we let the system start from the initial state $x_0$ and assume that there is a known  reference command $v^0$ such that the pair $(v^0,x_0)$ is safe. Moreover, the desired setpoint to which the system should be stabilized is $r\in  \mathcal R(\mathcal Y,\epsilon)$, i.e., the associated steady-output is strictly admissible with margin $\epsilon \in \mathbb R_{>0}$.
We now introduce the three elements involved in our supervisory scheme: safe set generation, path generation and reference governor implementation.

\subsubsection{Safe-set generation}
The generation of the safe sets is based on the  MOAS for different elements of the constraint set. For it, we introduce the following family of sets:
\begin{equation}\label{eq:MathcalYji}
    \mathcal Y_{i,j} =\begin{cases}
            &\mathcal Y_i, \text{ if } i= j,\\
            &\emptyset, \text{ if }F_{i,j}=0,\\
            & {\tt WE}(\mathcal Y_i,\mathcal  Y_j), \text{ else}
    \end{cases},
\end{equation}
where $i,j \in\mathbb Z_{[1,n_{\mathcal Y}]}$.  In the case  where $i=j$ the set $\mathcal Y_{i,j}$ corresponds to one of the elements composing $\mathcal Y$. The sets $\mathcal Y_{i,j}$ with $i\neq j$ and $F_{i,j} = 1$ correspond to the weak extension of $\mathcal Y_i,\;\mathcal Y_j$. This is a suitable choice for transitioning from $\mathcal Y_i$ to $\mathcal Y_j$, as ${\tt WE}(\mathcal Y_i,\mathcal Y_j)$ is a subset of $\mathcal Y_i\cup \mathcal Y_j$, overlaps both elements and is convex; this is shown in Section \ref{sec:theoreticalProps}. 

A family of safe sets is then computed offline, leveraging the inner approximation of the MOAS discussed in Section \ref{sec:DTLTI_and_RG}: 
\begin{equation}
    \mathcal O_{i,j}(\mathcal Y,\epsilon) = \mathcal O(\mathcal Y_{i,j},\epsilon).
\end{equation}
It is crucial that $\epsilon$ be chosen such that transitions between the different sections are still possible, i.e., such that 
\begin{equation} \label{eq:overlapInSSAdmissibleRefs}
    \mathcal R(\mathcal Y_{i,j},\epsilon)\cap \mathcal R(\mathcal Y_{i},\epsilon) \neq \emptyset\; \forall i,j: F_{i,j} = 1.
\end{equation}
A sufficient condition is to choose $\epsilon$ strictly smaller than the radius of the smallest Chebychev ball inscribed in $\mathcal Y_{i,j}\cap \mathcal Y_{i}$ for all $F_{j,i}\neq 0$. Proposition \ref{prop:overlapIsNonEmpty} ensures that a strictly positive $\epsilon$ always exists.
\subsubsection{Path generation}
Let $l^{\tt s},l^{\tt e}\in \mathbb Z_{[1,n_{\mathcal Y}]}$ be the \textit{start} and \textit{end} indices, i.e., such that $(v^0,x_0)\in \mathcal O_{l^{\tt s},l^{\tt s}}$ and $ H_\infty r\in \mathcal Y_{l^{\tt e}}$. Using connectedness of the collection of sets we can generate an \textit{index path}: a sequence $\{l_i\}_{i=1}^{n_l}$ corresponding to a path between $l_1 = l^{\tt s}$ and $l_{n_l} = l^{\tt e}$. The most straightforward approach is to compute the  undirected graph associated with the adjacency matrix and generate the shortest path connecting the nodes $l^{\tt s}$ and $l^{\tt e}$ using, e.g., the Dijkstra algorithm \cite{dijkstra2022note}. This returns the path with the least number of transitions. Modifying the weights of the adjacency matrix to minimize other metrics is also possible. A directed graph that heuristically produced short traveled distances in our simulation examples is obtained by using the following weights: 
\begin{equation}\label{eq:altWeights}
    w_{i,j} = b + \text{dist }(\mathcal Y_{i,j}\cap \mathcal Y_j, \text{segment}(H_\infty v^0, H_\infty r)), \text{ if } F_{i,j}=1,
\end{equation}
where $b\in\mathbb R_{>0}$  is a weight parameter penalizing the number of elements in a path and the second term is the distance between the relevant polytope  and the segment connecting the initial and desired setpoints.

Once the index path $\{ l_i\}_{i=1}^{n_l}$ has been generated, the intermediate reference command sequence, $\{r_i\}_{i=1}^{2n_l-1}$, is determined as
\begin{align*}
    &r_{2n_l-1} = r\\
    &r_{2i} =\quad \arg\hspace{-1cm} \min\limits_{v\in\mathcal R(\mathcal Y_{l_{i+1},l_{i}}\cap \mathcal Y_{l_{i+1}})}\|v-r_{2i+1}\|,\;i=  1,\dots, n_l-1, \\
    &r_{2i-1} = \quad\arg\hspace{-.85cm}\min\limits_{v\in\mathcal R(\mathcal Y_{l_{i},l_{i+1}}\cap \mathcal Y_{l_{i}})}\|v-r_{2i}\|,\;i=  1,\dots, n_l-1
\end{align*}
The sequence $\{r_i\}_{i=1}^{2n_l-1}$ is constructed by placing intermediate setpoints in the transition area between subsequent sets along the index path. Two setpoints are allocated for each transition area, one on each side of the associated gate. This ensures proper tracking performance. Without the steady-output compliance property in Assumption \ref{ass:connectedSequence} the intermediate reference command sequence may not be properly defined.

\subsubsection{Reference governor}

Algorithm \ref{algo:1} summarizes the proposed approach and corresponds to the computations to be performed online, at each time instant. In summary, we first determine whether we are ready to update the current index along the index path, then determine what safe set and what intermediate setpoint to use, finally, we compute the applied reference command using a CG scheme.

\begin{algorithm}
\caption{Input generation at time instant $k\in\mathbb Z_{\geq 0}$. } \label{algo:1}
\begin{algorithmic}[1]
\Require{ $x_k$: the current state; $v_{k-1}$: the  reference command at time $k-1$ with $v_{-1} = v^0$; $i_{k-1}\in\mathbb Z_{[1,n_l]}$: the location along the index path used at time $k-1$ (default to $1$).}
\State $i_k = i_{k-1}$
\If{$i_{k-1}<n_l$}
\If{$(v_{k-1},x_k)\in \mathcal O_{ l_{i_{k}},l_{i_k+1}}$}{ $i_k  = i_k+1$}
\EndIf
\EndIf
\State $\bar r_k = r_{2i_k -1}$, $\mathcal O_{k} = \mathcal O_{l_{i_k},l_{i_k}}$
\If{$(v_{k-1},x_k)\not\in\mathcal O_k$}{ $\mathcal O_{k} = \mathcal O_{l_{i_k-1},l_{i_k}}$,  $\bar r_k = r_{2i_k}$}
\EndIf
\State $v_k = \arg\min_{(v,x_k)\in\mathcal O_k}\|v-\bar r_k\|^2$
\end{algorithmic}
\end{algorithm}
Algorithm \ref{algo:1} proposes a unified CG approach irrespectively of whether the collection of polytopes is simply or strictly connected. In reality, if the collection is strictly connected one may use the extensions, reduce the number of intermediate setpoints by two and obtain possibly faster convergence to $r$.
\begin{rem}
    The proposed decomposition of the safe space  (weak extensions and intermediate setpoints) is also amenable to developing MPC-based control schemes. Feasibility guarantees may be obtained through chaining conditions applied to backward reachable sets corresponding to the intermediate setpoints. The details require careful treatment and are left as a topic for future research. Notably, it is expected that computing the MOAS, as we propose here, is much easier than the backward reachable sets. 
\end{rem}
\begin{rem}
    By construction, the proposed approach requires that the system trajectory enter the MOAS of each element and weak extension in the path. This can lead to slower convergence to the desired setpoint, as a constraint admissible trajectory that does not cross the individual MOAS might exist. This potential decrease in performance is offset by the scalability of our approach (see result in Section \ref{sec:compTimes}) and fast online computations (see times reported in Section \ref{sec:droneInCity}). Future work will explore ways to relax this requirement using, e.g., transient bridging controllers \cite{miller2000control}.
\end{rem}
As we show in the sequel, this approach has provable safety (Theorem \ref{thm:recursiveFeas}) and finite-time convergence guarantees (Theorem \ref{thm:finiteTimeConvToSteadyState}). Moreover, we show that computing the (weak) extensions is a straightforward  process (Propositions \ref{prop:HRepStrict}-\ref{prop:HRepLoose}) and, in some cases, the complexity of computing the  family of MOAS can be greatly reduced thereby making the method scalable to large connected collections (see Section \ref{sec:hyperrectangles}).

\section{Theoretical analysis}\label{sec:theoreticalProps}
\subsection{Theoretical properties of the (weak) extension}\label{sec:propsSequence}
We start by studying the (weak) extensions generated from a connected collection of polytopes. The key takeaways being that under suitable conditions the (weak) extension is a polytope (Proposition \ref{prop:HRepStrict}, Proposition \ref{prop:HRepLoose}) and that the intersection of the (weak) extensions with either of the constituent polytopes is full-dimensional (Proposition \ref{prop:overlapIsNonEmpty}). These results ensure the different MOAS can be computed efficiently and are nonempty.
Our first result relates to compactness of a connected collection.
\begin{lem} \label{lem:compact}
    Let ${\mathcal Y_1,\;\mathcal Y_2}$ form a simply connected collection of polytopes. If the elements of the connected collection are  compact, so is the extension of $\mathcal Y_1$ towards $\mathcal Y_2$.
\end{lem}
\begin{proof}
We first note that $\mathcal Y^{\tt res}_{2,1}$ is the intersection between a closed set and a compact set and therefore compact. $\mathcal Y^{\tt e}_{1,2}$ is the union of two compact sets and thus compact.
\end{proof}
Next, we construct a closed form expression of the extensions for a strictly connected collection. 
\begin{prop}\label{prop:HRepStrict}
    Let $\{\mathcal Y_1,\mathcal Y_2\}$ form a strictly connected collection of polytopes. Then, $\mathcal Y^{\tt e}_{1,2}$ is a polytope and an H-representation of $\mathcal Y^{\tt e}_{1,2}$ is given by the pair
    \begin{equation}\label{eq:Ext_HRep}
        \left(\begin{bmatrix}
            A^{\mathcal Y^{\tt o}_{1,2}}\\
            A^{(\mathcal Y^{\tt res}_{2})^{\tt o}_{2,1}}
        \end{bmatrix}, \begin{bmatrix}
            b^{\mathcal Y^{\tt o}_{1,2}}\\
            b^{(\mathcal Y^{\tt res}_{2})^{\tt o}_{2,1}}
        \end{bmatrix} \right),
    \end{equation}
    where, $(\mathcal Y^{\tt res}_{2,1})^{\tt o}_{2,1}$ is the opening of $\mathcal Y^{\tt res}_{2,1}$ towards $\mathcal Y^{\tt res}_{1,2}$. 
\end{prop}
\begin{proof}
 Let $\tilde {\mathcal Y}$ be the polytope defined by the pair in \eqref{eq:Ext_HRep}. The inclusion $\tilde {\mathcal Y}\subseteq \mathcal Y^{\tt e}_{1,2}$ holds directly as $y\in \tilde {\mathcal Y}$ and $A^{\mathcal Y_1}_{(G_{1,2})} y\leq b^{\mathcal Y_1}_{(G_{1,2})}\implies y\in\mathcal Y_1\subseteq \mathcal Y^{\tt e}_{1,2}$ and $y\in \tilde {\mathcal Y}$ and $A^{\mathcal Y_1}_{(G_{1,2})} y\geq b^{\mathcal Y_1}_{(G_{1,2})}\implies y\in\mathcal Y^{o}_{1,2}\cap\mathcal Y_2\subseteq \mathcal Y^{\tt e}_{1,2}$.\\
We now show that the inclusion $\mathcal Y^{\tt e}_{1,2}\subseteq \tilde {\mathcal Y}$ holds. By definition $y\in \mathcal Y^{\tt res}_{2,1}$ implies $y\in \tilde Y$. It remains to show that for all $y\in\mathcal Y_1$, $y\in(\mathcal Y^{\tt res}_{2,1})^{\tt o}_{2,1}$. Assume there exists $k\in Z_{[1,n_{\mathcal Y_2}]}$ with $k\neq G_{2,1}$ such that there is $y_1\in\mathcal Y_1$ $A^{\mathcal Y_2}_{(k)} y_1>b^{\mathcal Y_2}_{(k)}$ it is sufficient to demonstrate that $(A^{\mathcal Y_2}_{(k)},b^{\mathcal Y_2}_{(k)})$ does not represent a facet of  $\mathcal Y^{\tt res}_{2,1}$. By the strict connectedness property, it holds that for all $y\in\mathcal Y_1\cap {\tt H}(A^{\mathcal Y_2}_{(G_{2,1})},b^{\mathcal Y_2}_{(G_{2,1})}
)$, $A^{\mathcal Y_2}_{(k)} y\leq b^{\mathcal Y_2}_{(k)}$. We first show that: 
\begin{equation}\label{eq:proofHrepStrict}
    \forall y\in\mathcal Y^{\tt o}_{1,2}: A^{\mathcal Y_2}_{(G_{2,1})} y \leq b^{\mathcal Y_2}_{(G_{2,1})} \implies A^{\mathcal Y_2}_{(k)} y \leq b^{\mathcal Y_2}_{(k)}.
\end{equation}
Indeed, assume the statement does not hold and let $y_2$ be such an element. Then, there is $y_3(\kappa) = \kappa y_2 + (1-\kappa)y_1$ with $\kappa \in \mathbb R_{[0,1]}$ such that  $A^{\mathcal Y_2}_{(G_{2,1})}y_3(\kappa) = b^{\mathcal Y_2}_{(G_{2,1})}$ and $A^{\mathcal Y_2}_{(k)} y > b^{\mathcal Y_2}_{(k)}$. This implies that $y_3(\kappa)\not \in\mathcal Y^{\tt o}_{1,2}$ which is a contradiction as $\mathcal Y^{\tt o}_{1,2}$ is convex. Leveraging \eqref{eq:proofHrepStrict} it holds that 
\begin{align*}
    \mathcal Y^{\tt res}_{2,1} &= \mathcal Y^{\tt o}_{1,2}\cap {\tt P}(A^{\mathcal Y_2}_{(-k)},b^{\mathcal Y_2}_{(-k)})\cap {\tt P}(A^{\mathcal Y_2}_{(k)},b^{\mathcal Y_2}_{(k)})\\
    &= (\mathcal Y^{\tt o}_{1,2}\cap {\tt P}(A^{\mathcal Y_2}_{(k)},b^{\mathcal Y_2}_{(k)}))\cap {\tt P}(A^{\mathcal Y_2}_{(-k)},b^{\mathcal Y_2}_{(-k)})\\
    & = \mathcal Y^{\tt o}_{1,2}\cap {\tt P}(A^{\mathcal Y_2}_{(-k)},b^{\mathcal Y_2}_{(-k)}).
\end{align*}
Thus, $\mathcal Y_1\subseteq(\mathcal Y^{\tt res}_{2,1})^{\tt o}_{2,1}$, concluding the proof.   
\end{proof}
While Proposition \ref{prop:HRepStrict} is attractive for a strictly connected collection it cannot be used in the more general case of a simply connected collection. Indeed,  the  extension of two elements of a simply connected sequence may or may not be a polytope (see $\mathcal Y_{2,3}^{\tt e}$ and $\mathcal Y_{3,2}^{\tt e}$ in Figure \ref{fig:exampleConnectedSeq}). 
Fortunately, the weak extension is also a polytope and an H-representation is available for it . Before we present this we introduce two preliminary results. The first of which establishes a connection between the weak extension and the extensions for a strictly connected collection of polytopes. 
\begin{lem}
    \label{lem:intersecOfExtIsWeak} Let $\{\mathcal Y_1,\mathcal Y_2\}$ form a strictly connected collection of polytopes. Then, the following identity holds
    \begin{equation*}
        \mathcal Y^{\tt e}_{1,2} \cap \mathcal Y^{\tt e}_{2,1} = {\tt WE}(\mathcal Y_1,\mathcal Y_2).
    \end{equation*}
\end{lem}
\begin{proof}
    The proof follows from associativity and distributivity of unions and intersections as well as the fact $\mathcal Y_1\cap \mathcal Y_2\subseteq \mathcal Y^{\tt res}_{2,1}$ and $\mathcal Y_1\cap \mathcal Y_2\subseteq \mathcal Y^{\tt res}_{1,2} $ for a striclty connected collection. Indeed 
    \begin{align*}
        &\mathcal Y^{\tt e}_{1,2} \cap \mathcal Y^{\tt e}_{2,1} = (\mathcal Y_1 \cup \mathcal Y^{\tt res}_{2,1} )\cap (\mathcal Y_2 \cup \mathcal Y^{\tt res}_{1,2} )\\
        =& (\mathcal Y_1 \cap (\mathcal Y_2 \cup \mathcal Y^{\tt res}_{1,2})) \cup (\mathcal Y^{\tt res}_{2,1}\cap (\mathcal Y_2 \cup \mathcal Y^{\tt res}_{1,2}))\\
        =& ((\mathcal Y_1  \cap \mathcal Y_2) \cup \mathcal Y^{\tt res}_{1,2})\cup (\mathcal Y^{\tt res}_{2,1}\cap (\mathcal Y_2 \cup \mathcal Y^{\tt res}_{1,2}))\\
        =& \mathcal Y^{\tt res}_{1,2}  \cup(\mathcal Y^{\tt res}_{2,1} \cap (\mathcal Y_2 \cup \mathcal Y^{\tt res}_{1,2} ))\\
        =& \mathcal Y^{\tt res}_{1,2}  \cup(\mathcal Y^{\tt res}_{2,1}  \cup (\mathcal Y_2\cap\mathcal Y_1))\\
        =& \mathcal Y^{\tt res}_{1,2}  \cup\mathcal Y^{\tt res}_{2,1}  = {\tt WE}(\mathcal Y_1,\mathcal Y_2).
    \end{align*}
\end{proof}
Lemma \ref{lem:intersecOfExtIsWeak} directly implies that for a strictly connected pair the weak extension is a subset of the extension. Next, we show that for a simply connected pair of polytopes, the weak extension is equivalent to the extension of a strictly connected pair of subsets.
\begin{figure*}
\centering
\begin{subfigure}{.45\textwidth}
  \centering
 \includegraphics[width=1\linewidth]{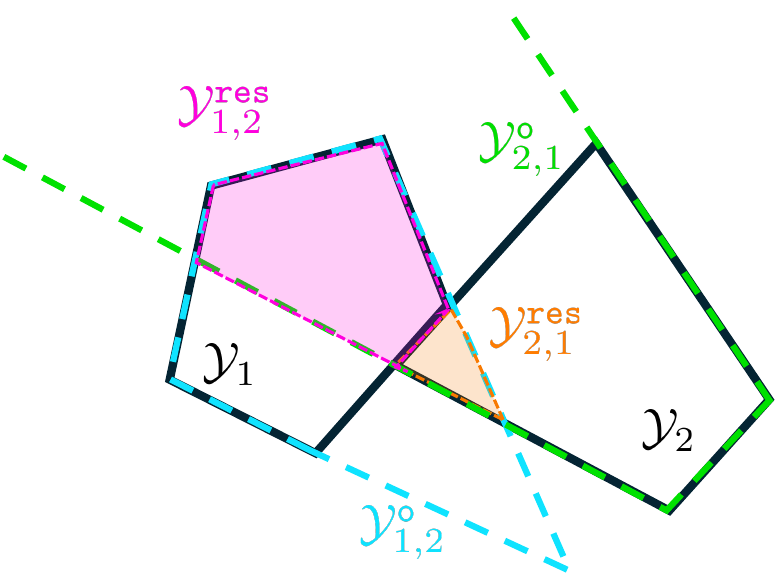}
    \caption{{ Connected collection of polytopes with $n_{\mathcal Y}=2$ also shown are relevant openings and restrictions.}}
    \label{fig:weakExt1}
\end{subfigure}%
$\;$
\begin{subfigure}{.45\textwidth}
  \centering
  \includegraphics[width=1\linewidth]{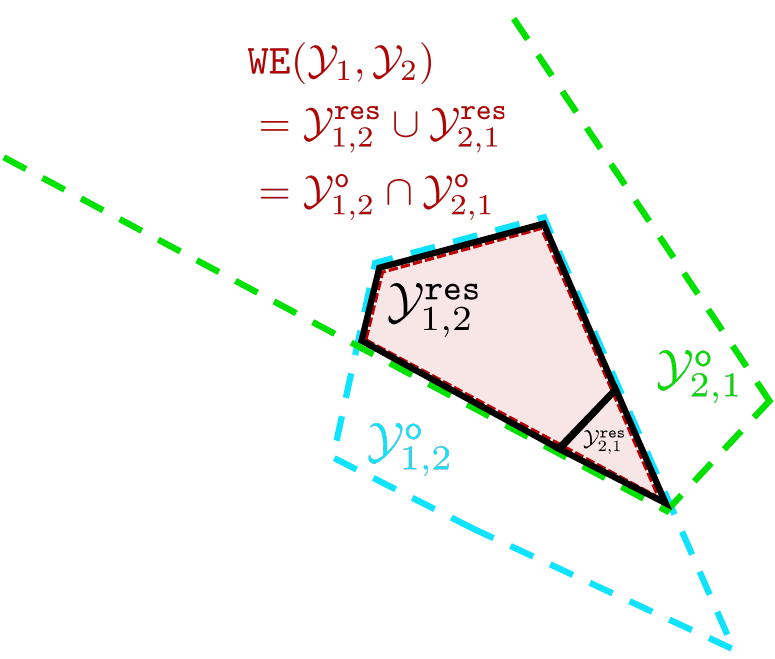}
  \caption{Generating the weak extension from the openings of the connected collection of polytopes.}
  \label{fig:weakExt2}
\end{subfigure}%
\caption{Illustration of the weak extension for a connected collection of polytopes}
\label{fig:exampleWeakExtension}
\end{figure*}
\begin{lem}\label{lem:weakExt}
     Let ${\mathcal Y_1,\;\mathcal Y_2}$ form a  simply connected collection of polytopes. Then, the restricted collection $\{\mathcal Y^{\tt res}_{1,2}, \mathcal Y^{\tt res}_{2,1}\}$ is a strictly connected collection and, moreover, 
     \begin{equation}\label{eq:identity1}
         {\tt WE}(\mathcal Y_1,\mathcal Y_2) = (\mathcal Y^{\tt res}_{1,2} )^{\tt e}_{1,2} = (\mathcal Y^{\tt res}_{2,1})^{\tt e}_{2,1}.
     \end{equation}
\vspace{.1cm} 
\end{lem}
\begin{proof}
    To show that the restricted collection is a strictly connected collection we show that the gate connecting $\mathcal Y_1,\;\mathcal Y_2$ has the required properties for strict connectedness. Indeed,  $\mathcal Y^{\tt res}_{1,2}  \cap {\tt H}(A^{\mathcal Y_2}_{(G_{2,1})},b^{\mathcal Y_2}_{(G_{2,1})}) = \mathcal Y_2\cap \mathcal Y_1 \cap {\tt H}(A^{\mathcal Y_2}_{(G_{2,1})},b^{\mathcal Y_2}_{(G_{2,1})}) = \mathcal Y^{\tt res}_{2,1}\cap {\tt H}(A^{\mathcal Y_2}_{(G_{2,1})},b^{\mathcal Y_2}_{(G_{2,1})})$, where we used the definition of openings and commutativity and associativity of the intersection operator. Combining this with connectedness of $\mathcal Y_1$ and $\mathcal Y_2$ we conclude that $\mathcal Y^{\tt res}_{1,2}$ and $\mathcal Y^{\tt res}_{2,1}$ are strictly connected. We now turn our attention to the first equality in the identity \eqref{eq:identity1}. 
    The right hand side is given by $$(\mathcal Y^{\tt res}_{1,2} )^{\tt e}_{1,2}=\mathcal Y^{\tt res}_{1,2}  \cup \left(\mathcal Y^{\tt res}_{2,1} \cap(\mathcal Y^{\tt res}_{1,2} )^{\tt o}_{1,2}\right).$$ The left hand side is $${\tt WE}(\mathcal Y_1,\mathcal Y_2) = \mathcal Y^{\tt res}_{1,2}  \cup \mathcal Y^{\tt res}_{2,1} .$$
    Then, the following condition is sufficient for the first equality in \eqref{eq:identity1} to hold:
    \begin{equation}\label{eq:identity1pt1}
        \mathcal Y^{\tt res}_{2,1} \subseteq (\mathcal Y^{\tt res}_{1,2} )^{\tt o}_{1,2}.
    \end{equation}
    This can easily be shown by noting that 
    \begin{align*}
        x\in \mathcal Y^{\tt res}_{2,1}  &\iff \begin{bmatrix}
            A^{\mathcal Y_1}_{(-G_{1,2})}\\
            A^{\mathcal Y_2}
        \end{bmatrix} x \leq \begin{bmatrix}
            b^{\mathcal Y_1}_{(-G_{1,2})}\\
            b^{\mathcal Y_2}
        \end{bmatrix},\\
        x\in (\mathcal Y^{\tt res}_{1,2} )^{\tt o}_{1,2} &\iff \begin{bmatrix}
            A^{\mathcal Y_1}_{(-G_{1,2})}\\
            A^{\mathcal Y_2}_{(-G_{2,1})}
        \end{bmatrix} x \leq \begin{bmatrix}
            b^{\mathcal Y_1}_{(-G_{1,2})}\\
            b^{\mathcal Y_2}_{(-G_{2,1})}
        \end{bmatrix},
    \end{align*}
    where we used the definitions of openings as well as the fact that the common hyperplane for the restriction is the same as the one for the original connected collection (shown in the first part of this proof). The proof for the second equality in \eqref{eq:identity1} holds by symmetry.
\end{proof}
We now show that the weak extension is a polytope.
\begin{prop}
    \label{prop:HRepLoose}
    Let ${\mathcal Y_1,\;\mathcal Y_2}$ form a simply connected collection of polytopes.  Then, ${\tt WE}(\mathcal Y_1,\mathcal Y_2)$ is a polytope. Moreover, the following identity holds,
\begin{equation}
    {\tt WE}(\mathcal Y_1,\mathcal Y_{2}) = \mathcal Y^{\tt o}_{1,2}\cap\mathcal Y^{\tt o}_{2,1}.
\end{equation}
\end{prop}
\begin{proof}
    The result holds directly from Proposition \ref{prop:HRepStrict} using Lemma \ref{lem:intersecOfExtIsWeak} in the case we have a strictly connected collection and Lemma \ref{lem:weakExt} in the case the collection is only simply connected.
\end{proof}
Figure \ref{fig:exampleWeakExtension} illustrates Proposition \ref{prop:HRepLoose} for a planar example.
\begin{rem}
    Proposition \ref{prop:HRepLoose} implies that the weak extension can be computed trivially, once the gate between the two polytopes has been identified. Moreover, it provides an alternative definition for the weak extension.
\end{rem}
Based on Proposition \ref{prop:HRepStrict} and Proposition \ref{prop:HRepLoose} we directly have the following corollaries\footnote{An alternative proof of Corollary \ref{cor:extConvexity} which does not depend on Proposition~\ref{prop:HRepStrict} is given in the Appendix}.
\begin{cor}
    \label{cor:extConvexity}
     Let ${\mathcal Y_1,\;\mathcal Y_2}$ form a strictly connected collection of polytopes. Then, the extension of $\mathcal Y_{1,2}^{\tt e}$ is convex.
\end{cor}
\begin{cor}\label{cor:weakerConvexity}
    Let ${\mathcal Y_1,\;\mathcal Y_2}$ form a  simply connected collection of polytopes. Then, ${\tt WE}(\mathcal Y_1,\mathcal Y_2)$ is convex. 
\end{cor}
 The final result of this section ensures that the intersection between the weak extension and either of the constituting polytopes is non-empty.
\begin{prop} \label{prop:overlapIsNonEmpty}
Let $\{\mathcal Y_1,\mathcal Y_2\}$ form a  simply connected collection of polytopes, then ${\tt WE}(\mathcal Y_1,\mathcal Y_2)\cap \mathcal Y_2 = \mathcal Y^{\tt res}_{2,1}$ and moreover $Y^{\tt res}_{2,1}$ has non-empty interior.
\end{prop}
\begin{proof}
First, let us show that $(\mathcal Y_1,\mathcal Y_2)\cap \mathcal Y_2= \mathcal Y^{\tt res}_{2,1}$. Indeed, from Proposition \ref{prop:HRepLoose} we have that 
${\tt WE}(\mathcal Y_1,\mathcal Y_2)\cap \mathcal Y_2 = \mathcal Y_1^{\tt o}\cap \mathcal Y_2^{\tt o}\cap \mathcal Y_2 = \mathcal Y^{\tt o}_{1}\cap \mathcal Y_2 = \mathcal Y^{\tt res}_{2,1}$, where we additionally used the definition of opening and restrictions. We now show that $Y^{\tt res}_{2,1}$ has non empty interior. 
For the rest of the proof we use the shorthand notation $\mathcal H_{2,1}$ for the hyperplane connecting both polytopes, i.e., ${\mathcal H}_{2,1} \triangleq {\tt H}(A^{\mathcal Y_2}_{G_{2,1}},b^{\mathcal Y_2}_{G_{2,1}})$.
     By construction and definition of a connected collection of polytopes $\mathcal Y^{\tt res}_{2,1}\cap \mathcal H_{2,1}$ is $n-1$ dimensional. Therefore, we can find $a\in\mathcal Y^{\tt res}_{2,1} \cap \mathcal H_{2,1}$ and $\epsilon>0$ such that $A^{\mathcal Y_2}_{(-G_{2,1})}a - b^{\mathcal Y_2}_{(-G_{2,1})}\leq-\epsilon$, as $a$ is an element strictly inside the $G_{2,1}^{\text{th}}$ facet of $\mathcal Y_{2}$. Then, for any $x\in\{a\}\oplus\frac{\epsilon}{\|A^{\mathcal Y_2}_{(-G_{2,1})}\|}\mathcal B$ and $j\in\mathbb Z_{[1,n_{\mathcal Y_2}]}\backslash\{G_{2,1}\}$, 
    \begin{align*}
        &A^{\mathcal Y_2}_{(j)} x - b^{\mathcal Y_2}_{(j)} = A^{\mathcal Y_2}_{(j)} a-b^{\mathcal Y_2}_{(j)} + \frac{\epsilon}{\|A^{\mathcal Y_2}_{-G_{2,1}}\|} A^{\mathcal Y_2}_{(j)} v, \text{ where }\|v\|\leq 1,\\
        &\leq -\epsilon + \frac{\epsilon}{\|A^{\mathcal Y_2}_{(-G_{2,1})}\|} A^{\mathcal Y_2}_{(j)} v \leq  -\epsilon + \epsilon\frac{\|A^{\mathcal Y_2}_{(j)}\|}{\|A^{\mathcal Y_2}_{(-G_{2,1})}\|} \leq \epsilon-\epsilon = 0.
    \end{align*}
    Now, define $\tau = \frac{\epsilon}{2\|A^{\mathcal Y_2}_{(-G_{2,1})}\|}$ and $b = a - \tau \frac{(A^{\mathcal Y_2}_{(G_{2,1})})^\top}{\|A^{\mathcal Y_2}_{(G_{2,1})}\|}$, clearly $\{b\}\oplus\tau\mathcal B\subset \{a\}\oplus\frac{\epsilon}{\|A^{\mathcal Y_2}_{(-G_{2,1})}\|}\mathcal B$ and moreover for all $x\in\{b\}\oplus \tau\mathcal B$ $A^{\mathcal Y_2}_{(G_{2,1})}x \leq b^{\mathcal Y_2}_{(G_{2,1})}$. Therefore we have $\{b\}\oplus\tau\mathcal B\subseteq\mathcal Y^{\tt res}_{2,1}$.
\end{proof}
A similar result to Proposition \ref{prop:overlapIsNonEmpty} can be obtained for the extension in the case of a strictly connected sequence by invoking Lemma \ref{lem:intersecOfExtIsWeak}, i.e., the weak extension is a subset of either extensions.
\subsection{Theoretical properties of the proposed Reference Governor}\label{sec:propsRG}
We are now ready to derive theoretical guarantees for the proposed supervisory control scheme. The first result shows constraint admissibility of the closed loop trajectories.
\begin{thm}
    \label{thm:recursiveFeas} Let Assumptions  \ref{ass:schur} and  \ref{ass:connectedSequence} hold. Given a pair $(v^0,x_0)\in \mathcal O_{l^0,l^0}$ for some $l^0 \in \mathbb Z_{[1,n_{\mathcal Y}]}$ the output trajectory obtained from $x_0$ through the dynamics \eqref{eq:dyns} in combination with the input from Algorithm \ref{algo:1} is constraint admissible at all time instants, i.e., $y_k\in \cup_{i=1}^{n_{\mathcal Y}} \mathcal Y_i$ for all $ k\in\mathbb Z_{\geq 1}$.
\end{thm}
\begin{proof}
First we note that for any $k\in\mathbb Z_{\geq 1}$, either $\mathcal O_k\neq\mathcal O_{k-1}$ or $\mathcal O_k=\mathcal O_{k-1}$. In the former case, lines 3 and 5 of Algorithm \ref{algo:1} ensure $(v_{k-1},x_k)\in\mathcal O_k$. In the latter case, leveraging the forward invariance property of the MOAS, it holds that $(v_{k-1},x_k)\in\mathcal O_k$. Therefore, $v_{k-1}$ is always a feasible solution to the optimization problem on Line 6. By definition of the MOAS (constraint satisfaction), safety at time $k$ is then ensured.
\end{proof} 
We now study the finite-time convergence properties of Algorithm \ref{algo:1}. Before we do so, we establish a couple of intermediate results ensuring that given an initial safe pair we can reach the next set in finite time. The first result states that from any safe set in a connected collection we can reach any weak extension generated by it in finite time. 
\begin{lem}\label{lem:fromSet2WE}
Let Assumptions \ref{ass:schur} and \ref{ass:connectedSequence} hold. Consider the dynamics \eqref{eq:dyns} and a constraint set $\{\mathcal Y_i\}_{i=1}^{n_{\mathcal Y}}$ representing a connected collection of polytopes. Take $i,j\in\mathbb Z_{[1,n_{\mathcal Y}]}$ such that $F_{i,j}=1$ and $\epsilon,\tau>0$ such that $\mathcal R(\mathcal Y^{\tt o}_{j,i}\cap\mathcal Y_i,\epsilon+\tau) \neq \emptyset$.  Given an initial condition $(v^0,x_0)\in\mathcal O(\mathcal Y_i,\epsilon)$, a desired setpoint $r\in\mathcal R(\mathcal Y^{\tt o}_{j,i}\cap\mathcal Y_i,\epsilon+\tau)$ and using the command governor associated with $\mathcal O(\mathcal Y_i,\epsilon)$ for input generation,
then, there exists $\ell\in\mathbb Z_{>0}$ such that 
$(v_k,x_k)\in\mathcal O({\tt WE}(\mathcal Y_i,\mathcal Y_j),\epsilon)$ for all $k\geq \ell$.
\end{lem}
\begin{proof}
   By definition $r\in\mathcal R(\mathcal Y_i,\epsilon)$ and leveraging convexity of $\mathcal Y_i$, Theorem \ref{thm:convToStrictSteadyState} ensures finite-time convergence of the applied reference command to $r$. Therefore, there is $\ell_0 \in \mathbb Z_{\geq 0}$ such that for all $k\geq\ell_0$ $v_k=r$. Moreover, $(r,x_k)\in\mathcal O(\mathcal Y_i,\epsilon)$ by forward invariance of the MOAS. Then, asymptotic stability of  \eqref{eq:dyns} ensures there exists $\ell>\ell_0$ such that $x_k \in \{H_\infty r\}\oplus \tau\mathcal B$ for all $k\geq \ell_1$ which directly implies $(r,x_k)\in \mathcal O(\mathcal Y^{\tt o}_{j,i}\cap \mathcal Y_i,\epsilon)$ for all $k\geq \ell$.
\end{proof}
The following result  states that starting from a safe pair belonging to a given weak extension we can reach setpoints in either of its constituents in finite time. 
\begin{lem}\label{lem:fromWE2Set}
Let Assumptions \ref{ass:schur} and \ref{ass:connectedSequence} hold.
Consider the dynamics \eqref{eq:dyns} and a constraint set $\{\mathcal Y_i\}_{i=1}^{n_{\mathcal Y}}$ representing a connected collection of polytopes. Take $i,j\in\mathbb Z_{[1,n_{\mathcal Y}]}$ such that $F_{i,j}=1$ and $\epsilon,\tau>0$ such that $\mathcal R(\mathcal Y^{\tt o}_{j,i}\cap\mathcal Y_i,\epsilon+\tau) \neq \emptyset$.  Given an initial condition  $(v^0,x_0)\in\mathcal O({\tt WE}(\mathcal Y_i,\mathcal Y_j),\epsilon)$,  a desired setpoint $r\in\mathcal R(\mathcal Y^{\tt o}_{j,i}\cap\mathcal Y_i,\epsilon+\tau)$ 
and using the command governor associated with $\mathcal O({\tt WE}(\mathcal Y_i,\mathcal Y_j),\epsilon)$,
then, there exists $\ell\in\mathbb Z_{>0}$ such that 
$(v_k,x_k)\in\mathcal O(\mathcal Y_i,\epsilon)$ for all $k\geq \ell$.
\end{lem}
The proof is similar to that of Lemma \ref{lem:fromSet2WE} and is omitted here. We note, however, that convexity of the weak extension (see Corollary \ref{cor:weakerConvexity}) is key in obtaining Lemma \ref{lem:fromWE2Set}, as else Theorem \ref{thm:convToStrictSteadyState} could not be used.
The following result states that using Algorithm \ref{algo:1}, and starting from a \textit{locally} safe reference command and initial condition pair the proposed scheme converges to any constant desired reference command in finite time.
\begin{thm}\label{thm:finiteTimeConvToSteadyState}
    Let Assumptions \ref{ass:schur} and \ref{ass:connectedSequence} hold. Consider the problem of bringing the dynamics \eqref{eq:dyns}, controlled using Algorithm \ref{algo:1} to the setpoint associated with $r\in\mathcal R(\mathcal Y,\epsilon)$, from the initial state $x_0$ subject to constraints \eqref{eq:cstr}, where $\mathcal Y = \cup_{i=1}^{n_{\mathcal Y}} \mathcal Y_i$ and $\{\mathcal Y_i\}_{i=1}^{n_{\mathcal Y}}$ is a connected collection of polytopes. Suppose 
    that there exists a reference command $ v^0$ such that $(v^0,x_0)\in\mathcal O_{l^0,l^0}$ for some $l^0 \in\mathbb Z_{[1,n_{\mathcal Y}]}$. Setting $v_0 = v^0$ ensures finite-time convergence of $v_k$ to $r$.
\end{thm}
\begin{proof} 
We first note that  under Assumptions \ref{ass:schur} and \ref{ass:connectedSequence}, Theorem \ref{thm:convToStrictSteadyState}
can be applied for any set $\mathcal O_{i,j}$, $i,j\in \mathbb Z_{n_{\mathcal Y}}$ such that $F_{i,j}= 1$.
Now, if $i_k = n_l$, then by  Lemma \ref{lem:fromWE2Set} there exists $\ell$ such that $(v_{\ell-1},x_\ell)\in \mathcal O_{l_{n_l},l_{n_l}}$. Under Assumptions \ref{ass:schur} and \ref{ass:connectedSequence}, Theorem \ref{thm:convToStrictSteadyState} then ensures finite-time convergence to $r$. If $i_k\neq n_l, i_k\neq 1$ Lemma \ref{lem:fromWE2Set} ensures that there is $\ell$ such that $(v_{\ell-1},x_\ell)\in \mathcal O_{l_{i_k},l_{i_k}}$. Then, Lemma \ref{lem:fromSet2WE} ensures there is $\ell$ such that $i_{\ell+1} = i_{\ell}+1$. Finally, if $i_k=1$ by forward invariance of the MOAS and assumption on $(v^0,x_0)$ Line 5 is never executed and  Lemma \ref{lem:fromSet2WE} ensures we have $i_k = 2$ in finite time.  As $n_l$ is finite, we can conclude that there is a finite $\ell$ such that $v_k = r$ for all $k\geq \ell$. throughout the proof we made use of the definition of the sequence $\{r_i\}_{i=1}^{2n_l-1}$ which is such that we can apply Lemmas \ref{lem:fromSet2WE}-\ref{lem:fromWE2Set}.
\end{proof}
\begin{rem}
    Algorithm \ref{algo:1} describes the scheme in the setting of the command governor but other RGs can be used without further modification to the scheme as long as they have convergence properties equivalent to those in Theorem \ref{thm:convToStrictSteadyState}. One such example is the scalar RG.
\end{rem}
\begin{rem}
    While the approach is presented here for disturbance free linear systems. The safe set considered in this work and general idea can also be combined with reference governors for systems with disturbance inputs, see \cite{kolmanovsky1998theory} as well as combined with Lyapunov function based reference governors for nonlinear systems, see \cite{garone2017reference} and references therein. The details of which are left for future study.
\end{rem}

\begin{figure*}
\centering
\begin{subfigure}{.45\textwidth}
  \centering
 \includegraphics[width=1\linewidth]{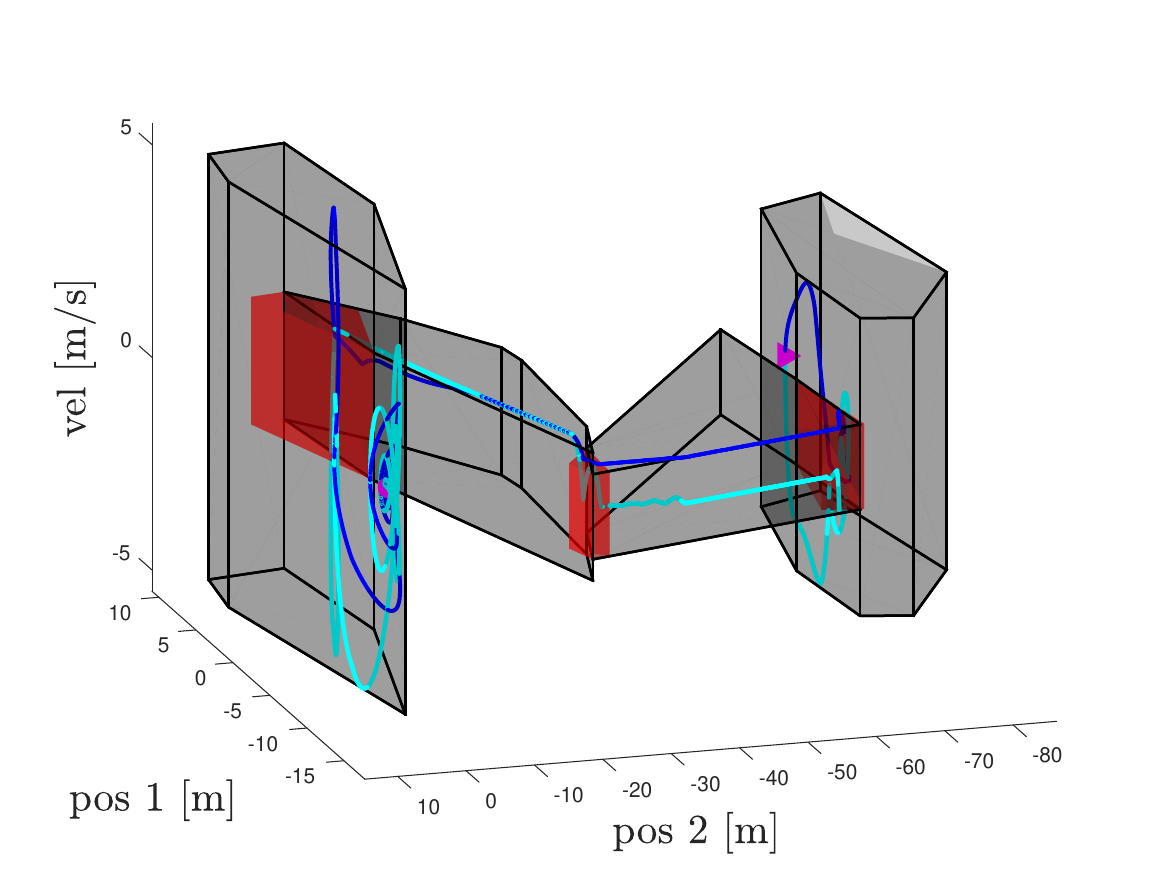}
    \caption{{\small In cyan, the trajectory considering position histories and first velocity component. In blue, the trajectory considering position histories and second velocity component.}}
    \label{fig:MSP_traj1}
\end{subfigure}%
$\;$
\begin{subfigure}{.45\textwidth}
  \centering
  \includegraphics[width=1\linewidth]{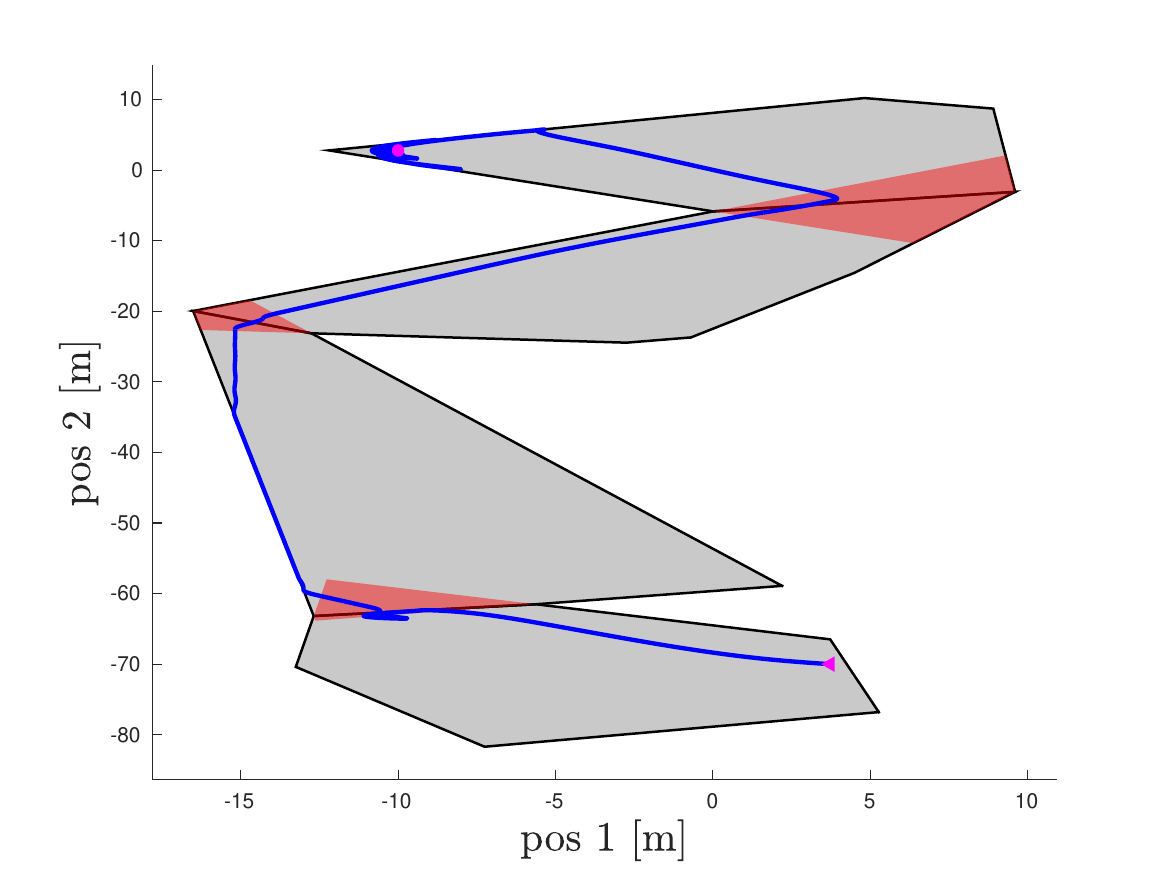}
  \caption{{\small Projection into position space of the 2D-MSD trajectory. Trajectory is depicted in blue.\\ $\;$ }}
  \label{fig:MSP_traj2}
\end{subfigure}%
\caption{Trajectory of the 2D-MSD inside a randomly generated sequence of polytopes in 4D space. A 3D representation (left) and 2D representation (right) are given. The sequence of polytopes is shown in transparent black, the weak extensions are shown in red, the initial condition and the desired set point  are represented by a magenta triangle and dot, respectively.}
\label{fig:MSP_trajs}
\end{figure*}
\subsection{Fast safe set computation for hyperrectangles}\label{sec:hyperrectangles}
When the number of elements in a connected collection, $n_\mathcal Y$, increases, the number of MOAS that need to be computed also grows, possibly leading to long offline computational times. 
Nevertheless, in certain cases,  properties of the MOAS can be used to substantially reduce the computational load. These properties are related to scaling and intersections of the MOAS.  We now recall said properties and then explain how we can use them.   
\begin{thm}\label{thm:scaling}
Let Assumptions \ref{ass:simpleCstr} and \ref{ass:schur} hold. \\
    (i)\cite[Theorem 2.1 iv)]{gilbert1991linear} Consider the dynamics \eqref{eq:dyns} and a feasible region $\mathcal Y$. Let $\alpha\in\mathbb R$, we then have that $\mathcal O(\alpha\mathcal Y) = \alpha\mathcal O(\mathcal Y)$.\\
    (ii)\cite{kolmanovsky1998theory}
    Let $\mathcal Y_1,\mathcal Y_2\in \mathbb R^{n_y}$ then $\mathcal O(\mathcal Y_1\cap\mathcal Y_2) = \mathcal O(\mathcal Y_1)\cap \mathcal O(\mathcal Y_2)$.
\end{thm}
 Theorem \ref{thm:scaling}(i) states that scaling of the constraint leads to scaling of the MOAS by the same factor.  Theorem \ref{thm:scaling} (ii) states that the MOAS of the intersection of two constraint sets is equal to the intersection of the individual MOAS. 
 Finally, we mention that for linear inequality constraints, a translation of the constraint set leads to a translation of the MOAS.

 The idea underlying the fast safe set computation is that if all the sets $\mathcal Y_{i,j}$, $i,j\in\mathbb Z_{[1,n_\mathcal Y]}$ are obtained through scaling and translation of a  \textit{base polytope} one would only need to compute the MOAS of the base polytope once and then apply trivial scaling and translation operations. It is, however, unlikely that all sets and weak extensions can be obtained by scaling and translation of a single or even a reduced number of base polytopes.
Take, for example, the case where the elements of the connected collection are all hyperrectangles,  it turns out that the weak extensions will also be hyperrectangles. Nevertheless, there is no \textit{base polytope} from which all hyperrectangles can be generated by scaling alone. However, as we see next, Theorem \ref{thm:scaling} (ii) alows us to circumvent this issue.

For the rest of this section we let the elements of the connected collection be given by $\mathcal Y_i = \mathcal Y^{\tt rect}_{i,i}\cap \mathcal Y^{\tt cstr}$ where $\mathcal Y^{\tt rect}_{i,i}$ represents a hyperrectangle and $\mathcal Y^{\tt cstr}$ represents a set of additional constraints on the outputs, for example in the case of a robot moving in 3D space, the constraints $\mathcal Y^{\tt cstr}$ may represent input saturation and velocity constraints among others, while the sets $\mathcal Y^{\tt rect}_i$ represent the obstacle free regions. Then, the weak extension ${\tt WE}(\mathcal Y_i,\mathcal Y_j) = \mathcal Y^{\tt rect}_{i,j}\cap \mathcal Y^{\tt cstr}$ where $\mathcal Y^{\tt rect}_{i,j}$ is also a hyperrectangle. As such, the sets $\mathcal Y^{\tt rect}_{i,j}$ are characterized by the vectors of upper and lower bounds $b^{\tt up}_{i,j},b^{\tt low}_{i,j}$.
Algorithm \ref{algo:2} details how to efficiently compute $\mathcal O(\mathcal Y_{i,j})$. 

\begin{algorithm}
\caption{Efficient computation of the MOAS for hyperrectangles. } \label{algo:2}
\begin{algorithmic}[1]
\Require{ $F$: adjacency matrix, $b^{\tt up}_{i,j},b^{\tt low}_{i,j}$: vectors of upper and lower bounds describing the different hyperrectangles, $\mathcal Y^{\tt cstr}$: additional constraints.}
\State Define unitary constraint sets on each component of the output vector: $\mathcal A_m = \{y\in\mathbb R^{n_y}:\,|(y_{(m)})|\leq 1\}, \;m = 1,\dots, n_y$.
\State Compute $\mathcal O(\mathcal Y^{\tt cstr},\epsilon)$ and $\mathcal O(\mathcal A_m,\epsilon),\;m = 1,\dots, n_y$. 
\For{$i = 1,\dots n_{\mathcal Y}$, and $j = 1,\dots n_{\mathcal Y}$}
\If{$F_{i,j}=1$}
\State $\alpha = (b^{\tt up}_{i,j}-b^{\tt low}_{i,j})/2$, $c =( b^{\tt up}_{i,j}+b^{\tt low}_{i,j})/2$,\State $r_c: c = H_\infty r_c$, $\mathcal O(\mathcal Y_{i,j},\epsilon) = \mathbb R^{n_y}$. 
\For{$m = 1,\dots,n_y$} 
\State $\mathcal O(\mathcal Y_{i,j},\epsilon) = \mathcal O(\mathcal Y_{i,j},\epsilon)\cap \alpha_{(m)}\mathcal O(\mathcal A_m,\epsilon)$
\EndFor
\State $\mathcal O(\mathcal Y_{i,j},\epsilon) = \mathcal O(\mathcal Y_{i,j},\epsilon) \oplus \{(r_c, c)\}$
\State $\mathcal O(\mathcal Y_{i,j},\epsilon) = \mathcal O(\mathcal Y_{i,j},\epsilon) \cap \mathcal O(\mathcal Y^{\tt cstr},\epsilon) $
\EndIf
\EndFor
\end{algorithmic}
\end{algorithm}

It should be noted that Algorithm \ref{algo:2} requires the center of each hyperrectangle, i.e. the temporary variable $c$, to be associated to a steady-state, $r_c$, as defined in Line 6. The approach is used in Sections \ref{sec:CWH},  Section \ref{sec:droneInCity} and insights on the computation time are given in  Section \ref{sec:compTimes}.
\section{Numerical simulations}\label{sec:numericalSim}
We now illustrate the applicability of the approach by considering various dynamical systems and constraint sets. Simulations are performed using Matlab\textregistered \@ on a Mac book Pro Laptop with Apple M2 Pro processor and 16 Gb of RAM. It is worth noting that both Assumption \ref{ass:schur} and Assumption \ref{ass:connectedSequence} hold for each of the cases considered hereunder.
The code for the simulations in Section \ref{sec:CWH} and Section \ref{sec:droneInCity} can be found at \url{https://github.com/mcastrov-pixel/Reference-Governor-for-union-of-polytopes}
\subsection{Reference tracking inside a collection of connected rooms}\label{sec:2DMSP}

In this subsection, we aim to illustrate the safe tracking capabilities of our approach when considering highly oscillating dynamics. Moreover, we also showcase the inherent ability of our approach to consider sequences of polytopes which cannot directly be expressed as geometric constraints.
To do so, we consider a planar example with dynamics described by two uncoupled second order systems as
\begin{equation}\label{eq:dyns_MSP}
    \ddot x^c_i = -\omega_{{\tt n}i}^2x^c_i -2\zeta_i\omega_{{\tt n}i}\dot x^c_i,\;i = 1,2,
\end{equation}
with natural frequencies $\omega_{{\tt n}1} = 2$, $\omega_{{\tt n}2} = 1$ and  damping ratios $\zeta_1 = 0.1$, $\zeta_2 = 0.08$. After discretization using a sampling period of $0.05$ s we obtain a DT LTI system with four states, the first two representing positions and the last two representing velocities along the respective axes. The steady-state associated  with the reference command, $v_k$, corresponds to position $v_k$ with zero velocity. 
Note that the system is highly under-damped. We refer to this system as a two dimensional Mass spring Damper system (2D-MSD).

We consider a tracking task while requiring the 2D-MSD to remain inside a connected collection of polytopes in the four dimensional state space. The projection of the collection into the position plane can be seen in Figure \ref{fig:MSP_traj2}, note that in the position plane the collection is strictly connected. Moreover, we also impose a maximum on the  infinity norm of the velocity, with a different saturation value for each element in the sequence of polytopes\footnote{this type of velocity constraint arises naturally in, e.g., automotive applications and also in the case of close proximity constraint \cite{castroviejo2024safe}}. This naturally translates to a simply connected of polytopes in the four dimensional state space. A projection of the sequence of polytopes in three dimensional space is given in Figure \ref{fig:MSP_traj1} where both velocity components have been collapsed into a single axis.

Figure \ref{fig:MSP_trajs} shows a typical trajectory for the 2D-MSD. We note that in agreement with Theorem \ref{thm:finiteTimeConvToSteadyState}, trajectories converge to the desired setpoint in finite time. Moreover, and despite the highly underdamped dynamics, this occurs without any constraint violations, as expected from Theorem \ref{thm:recursiveFeas}. Similar observations apply to the following examples. In Figure \ref{fig:MSP_traj1}, we can see that the MSP travels at the maximum allowed velocity most of the time. Moreover, the trajectory oscillates (due to the low dampening coefficients) whenever a more stringent constraint is to be met (and similarly for final convergence). In contrast, a transition that corresponds to a relaxation of the velocity constraint does not lead to a significant decrease in the velocity before transitioning. From Figure \ref{fig:MSP_traj2} we note that the lower most weak extension has a very small overlap with one of the polytopes. This is due to one of the facets being almost parallel to the gate. Despite this very challenging situation, the weak extension is nonempty (as derived in Proposition \ref{prop:overlapIsNonEmpty}) and the 2D-MSD is able to safely transition between the elements of the sequence.
\subsection{On-orbit proximity operation}
\label{sec:CWH}
Next, we consider a more complex 6 dimensional system with coupled dynamics which represents a  point-mass satellite performing a proximity maneuver around a Chief spacecraft in a 400 km altitude low earth circular orbit. The relative dynamics of the satellite with respect to the Chief spacecraft center of mass are given by the CWH equations see, for example, \cite[Equation (12)]{castroviejo2023} for more details concerning the dynamics. The state vector is then comprised of six states, the first three and last three states represent radial, along track and cross track positions and velocities, respectively. The inputs are relative accelerations (thrust forces divided by spacecraft mass) along the three axes, with component-wise input saturation limits of 0.1 m/s$^2$. A sampling period of 2 s is used. The satellite relative position is stabilized to a specified position using an LQR controller corresponding to $10^3$ weights for the position states, $10^{-1}$ weights for the velocity states and $10^{-2}$ for the inputs.

Figure \ref{fig:ISS_rep} illustrates the Chief spacecraft and the decomposition of the space around it. More precisely, the Chief spacecraft is approximated as a union of 23 hyperrectangles and the free space is divided into a union of non-overlapping hypercubes obtained using a naive greedy algorithm that scans the three dimensional space along the different directions. Moreover, the hyperrectangles in the free space that have either a smallest dimension smaller than 1.5 meters or both a smallest dimension smaller than 3 meters and a volume smaller than 175 m$^3$ are discarded as too small. Following this procedure, we obtain a total of 84 hyperrectangles forming a connected sequence.

We consider maneuvers between an initial position and a desired setpoint with the path generated using the distance-based weights described in \eqref{eq:altWeights} with tunning parameter $b = 0.1$. 
\begin{figure}[h]
  \centering
  \centering
 \includegraphics[width=1\linewidth]{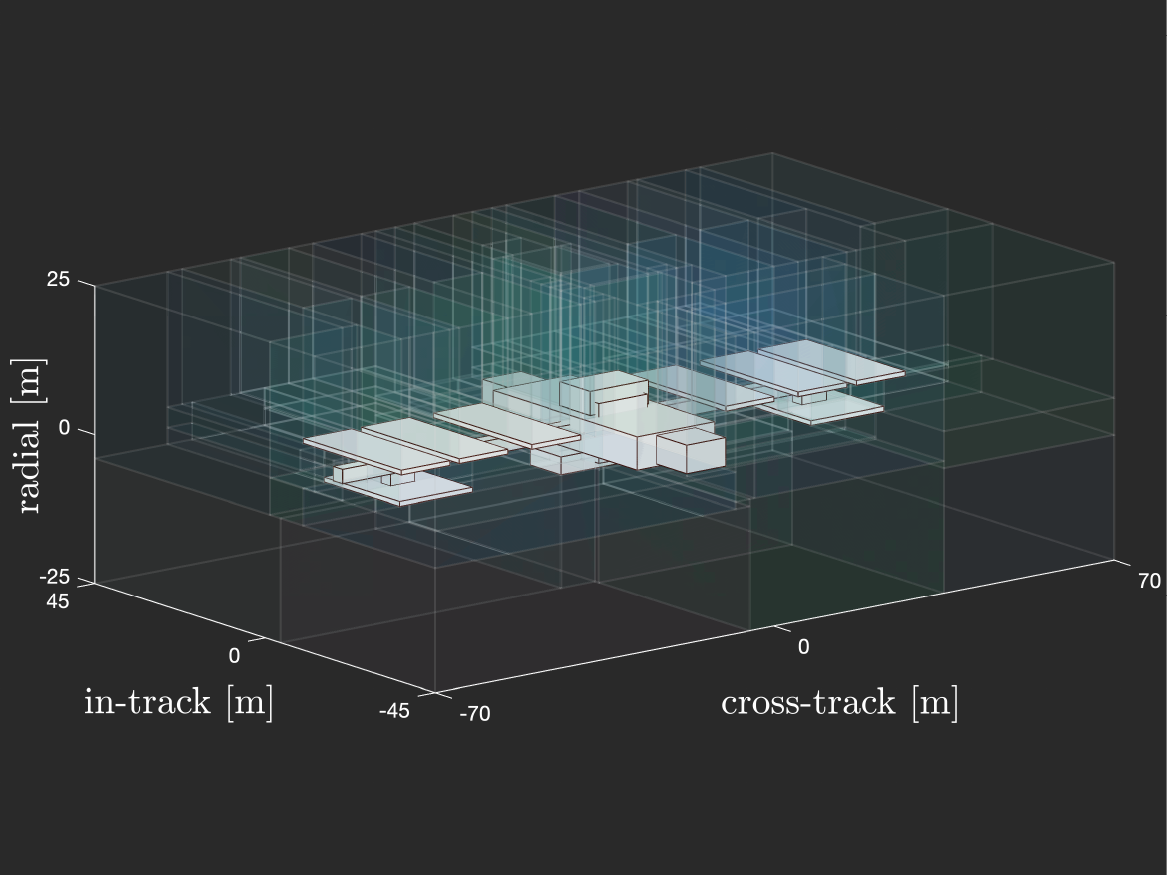}
    \caption{Representation of the spacecraft (gray) as well as the decomposition of the free space (transparent boxes).}
    \label{fig:ISS_rep}
\end{figure}%
 Figure \ref{fig:CWH3_traj} depicts different maneuvers. In both cases the satellite reaches the desired setpoint while avoiding collision with the Chief spacecraft. Both maneuvers are completed within 14 minutes while respecting the velocity and torque constraints.
\begin{figure*}
\centering
\begin{subfigure}{.48\textwidth}
  \centering
  \centering
 \includegraphics[width=1\linewidth]{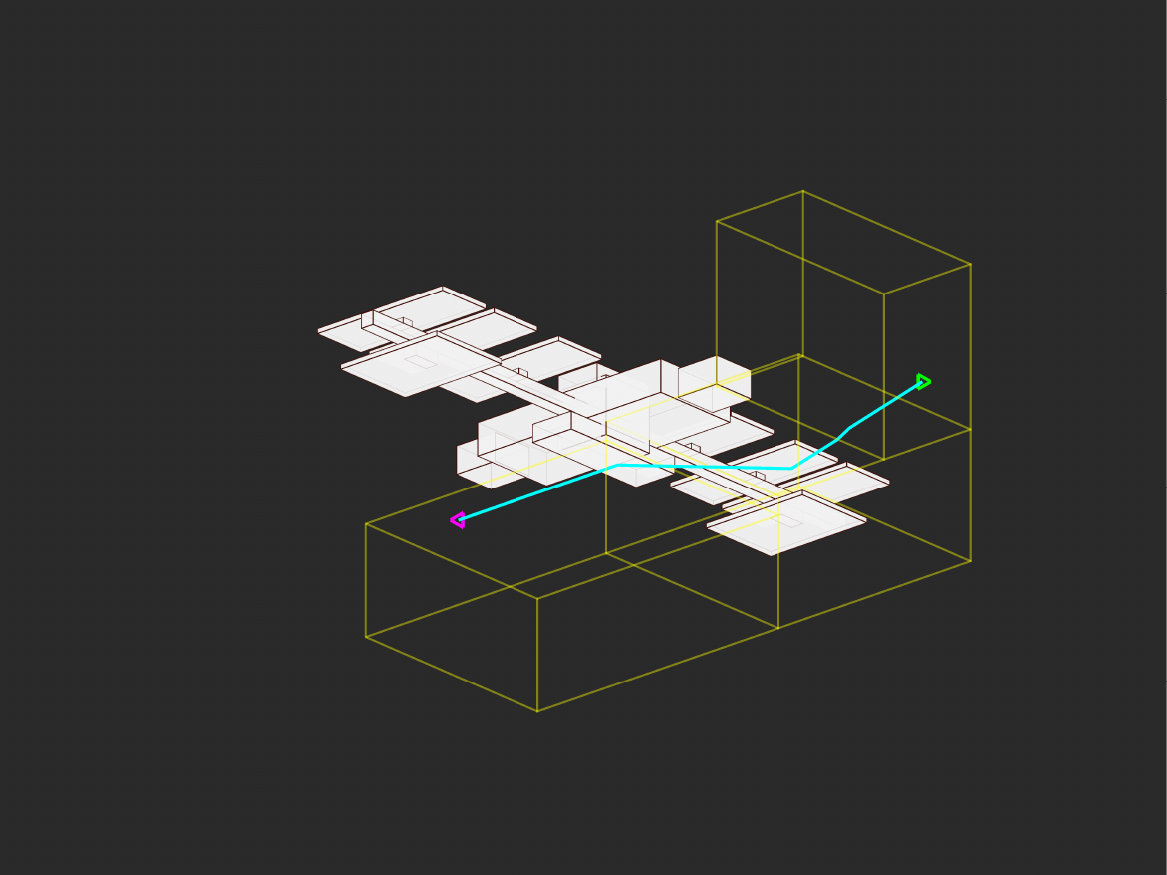}
    \caption{}
    \label{fig:CWH3_traja}
\end{subfigure}%
$\;$
\begin{subfigure}{.48\textwidth}
   \centering
 \includegraphics[width=1\linewidth]{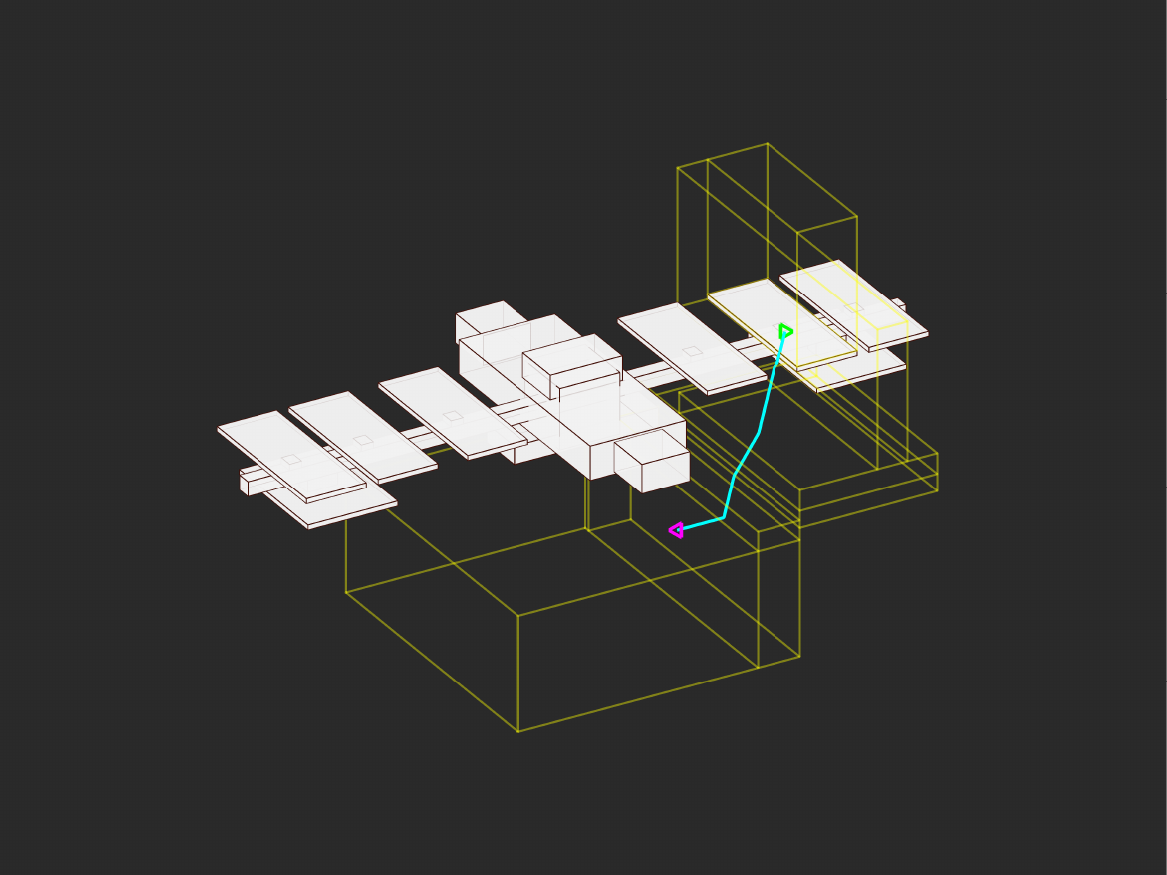}
    \caption{}
    \label{fig:CWH3_trajb}
\end{subfigure}%
\caption{Different on-orbit proximity maneuvers (left and right). The green and magenta triangles denote the satellite initial condition and desired setpoint, respectively. The path of the satellite is represented by a cyan curve. The yellow boxes are the sequence of polytopes that are traversed.}
\label{fig:CWH3_traj}
\end{figure*}
\subsection{Quadcopter control in urban environment}
\label{sec:droneInCity}
In this section we demonstrate the approach on a practically relevant nonlinear system that is feedback linearizable. 
More specifically, we consider the tracking problem for a quadcopter  operating in an urban environment. After exact linearization the translational dynamics can be represented as three double integrators with inputs corresponding to acceleration along the three major axis, see \cite[Section IV]{do2024lp} for a concise description of the procedure. Due to the feedback linearization, the box-type constraint on the original input signals (roll, pitch and positive vertical thrust) is translated into a coupled constraint on the three accelerations. This constraint is exactly represented by the intersection of a cone and a sphere and forms a convex set, see \cite[Equation (33)]{do2024lp}, parameterized by the maximum thrust, $T_{max}$, and the maximum  yaw/pitch angle, $\alpha_{max}$. Here we consider $T_{max} = 1.4 g$ where $g = 9.81$ m s$^{-2}$, and $\alpha_{max} = 15^\circ$. A tight polytopic inner approximation of the constraint set is obtained using the procedure outlined in \cite[Equation (35)]{do2024lp} and choosing 150 vertices. The resulting inner approximation of the set of constraints on the acceleration inputs is presented in Figure \ref{fig:drone_AccInput}. The difference in volume of the inner approximation with respect to the original constraint set normalized by the volume of the original set is 3.4\%. Due to the small relative error, visualizing the difference was challenging and is therefore omitted here.
The sampling frequency is 10 ms and the quadcopter is stabilized to a specified position using an LQR controller corresponding to weights of 10 for position states, $10^{-1}$ for velocity states and $10^{-2}$ for the acceleration inputs. Also, the infinity norm of the velocity vector is limited to 1 m s$^{-1}$.

The urban environment, depicted in Figure \ref{fig:dron3_traja}, is constructed by pseudo-randomly  generating 35 non-overlapping buildings of different heights, widths and depths. The free space is then decomposed using the same approach as in Section \ref{sec:CWH} but no restriction on the minimum size of the hyperrectangles was imposed, leading to a sequence of 141 connected hyperrectangles. Given an initial and final position, the path and intermediate references are computed as in Section \ref{sec:CWH} using a weight $b=0.1$.

Figure \ref{fig:dron3_traj} shows trajectories from 140 different initial conditions to a single desired reference point. None of the trajectories collided with any of the buildings, all trajectories were within one meter of the desired setpoint within 105 s and 138 of them were within $10^{-2}$ meters of the target within 105 s.
The mean and maximum times needed to generate the directed graph and associated path for the 140 simulations were 9 and 22 ms, respectively. The mean and maximum online computation time needed to solve the optimization problem associated with the command governor over all simulations and all time instants are 2.5 ms and 3 ms, respectively. These results suggest that the CG is attractive for applications with fast sampling times. This is not surprising as the  CG optimization problem has only 3 decision variables. A comparison with other constrained control approaches in terms of the computation time is out of the scope for this paper. Nevertheless, a comparison between a different RG scheme and MPC for a similar-sized problem can be found in  \cite[Section 5.2.4]{castroviejo2025robust}.

Figure \ref{fig:drone_AccInput} shows typical acceleration (input) trajectories as well as the constraint set. We note that the inputs are often on or near the constraint boundary indicating that the CG approach does not lead to excessive conservatism in this example. This observation is reinforced by examining the velocity time-histories, shown in Figure \ref{fig:drone_vels}, where the constraint boundary is often reached.
\begin{figure}[h]
  \centering
 \includegraphics[width=1\linewidth]{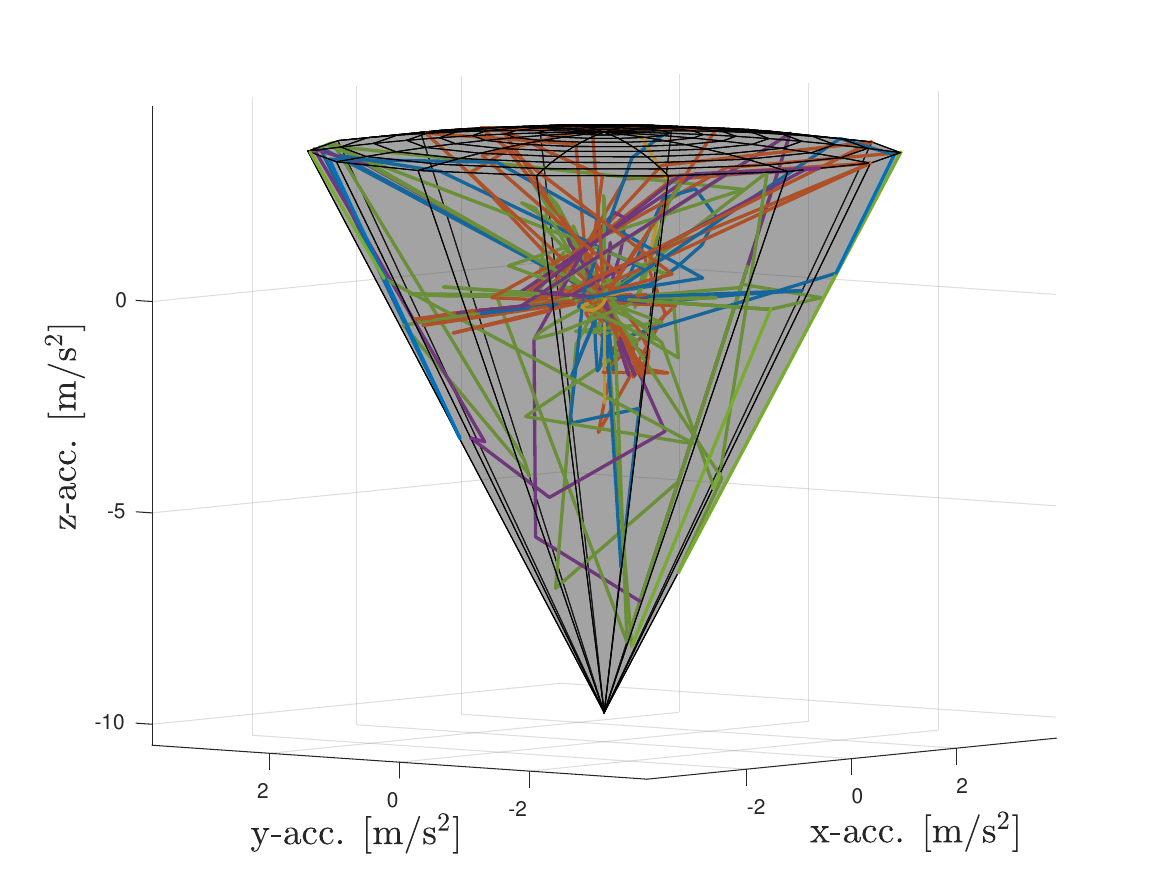}
    \caption{Typical Acceleration inputs histories (solid lines) for the quadcopter tracking task. Constraints on the input are approximated through a polytope made of 35 hyperplanes  (gray).}
    \label{fig:drone_AccInput}
\end{figure}
\begin{figure*}
\centering
\begin{subfigure}{.48\textwidth}
  \centering
  \centering
 \includegraphics[width=1\linewidth]{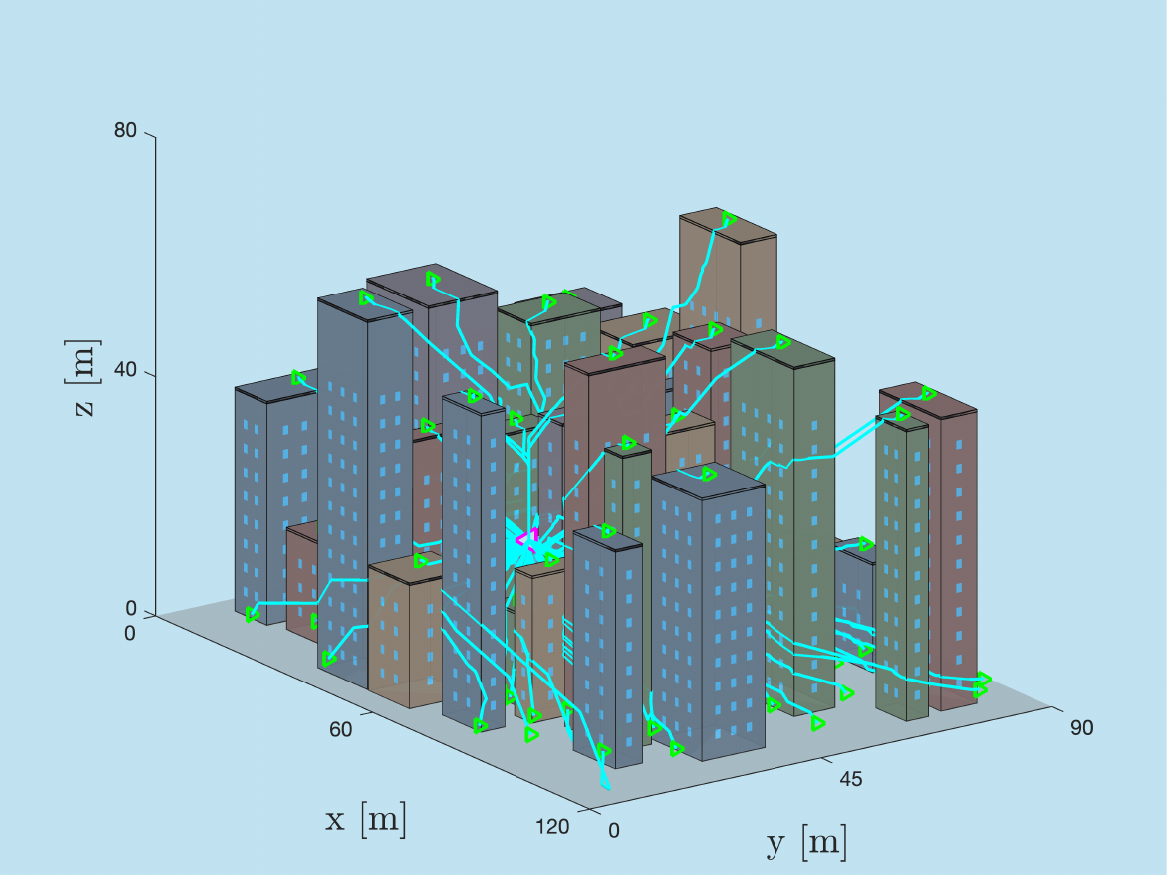}
    \caption{View from the side.}
    \label{fig:dron3_traja}
\end{subfigure}%
$\;$
\begin{subfigure}{.48\textwidth}
   \centering
 \includegraphics[width=1\linewidth]{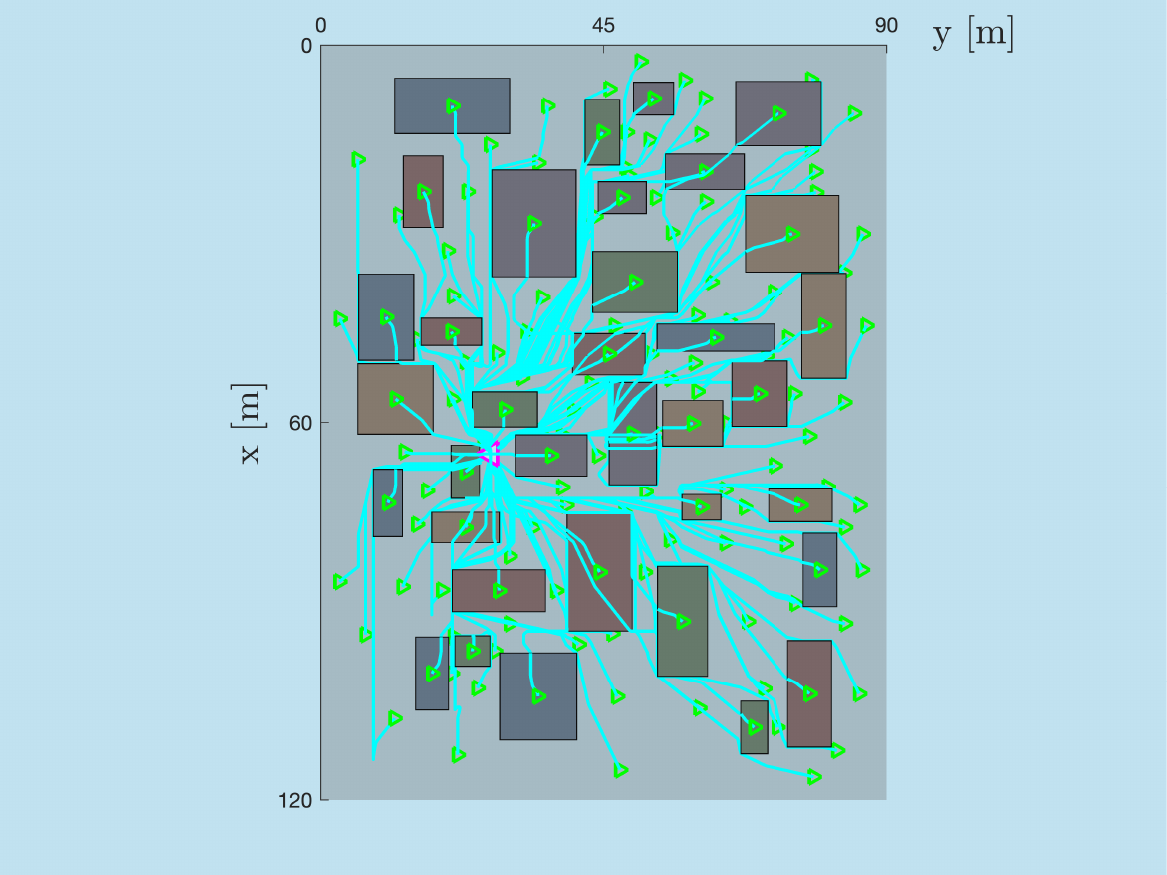}
    \caption{View from the top.}
    \label{fig:dron3_trajb}
\end{subfigure}%
\caption{Maneuvers for the quadcopter navigating to a given setpoint from different initial conditions. Left and right show different views of the same task. In both figures, the green and magenta triangles denote the quadcopter initial conditions and desired setpoint, respectively. The paths of the quadcopter are represented by cyan lines.}
\label{fig:dron3_traj}
\end{figure*}
\begin{figure}[h]
  \centering
 \includegraphics[width=1\linewidth]{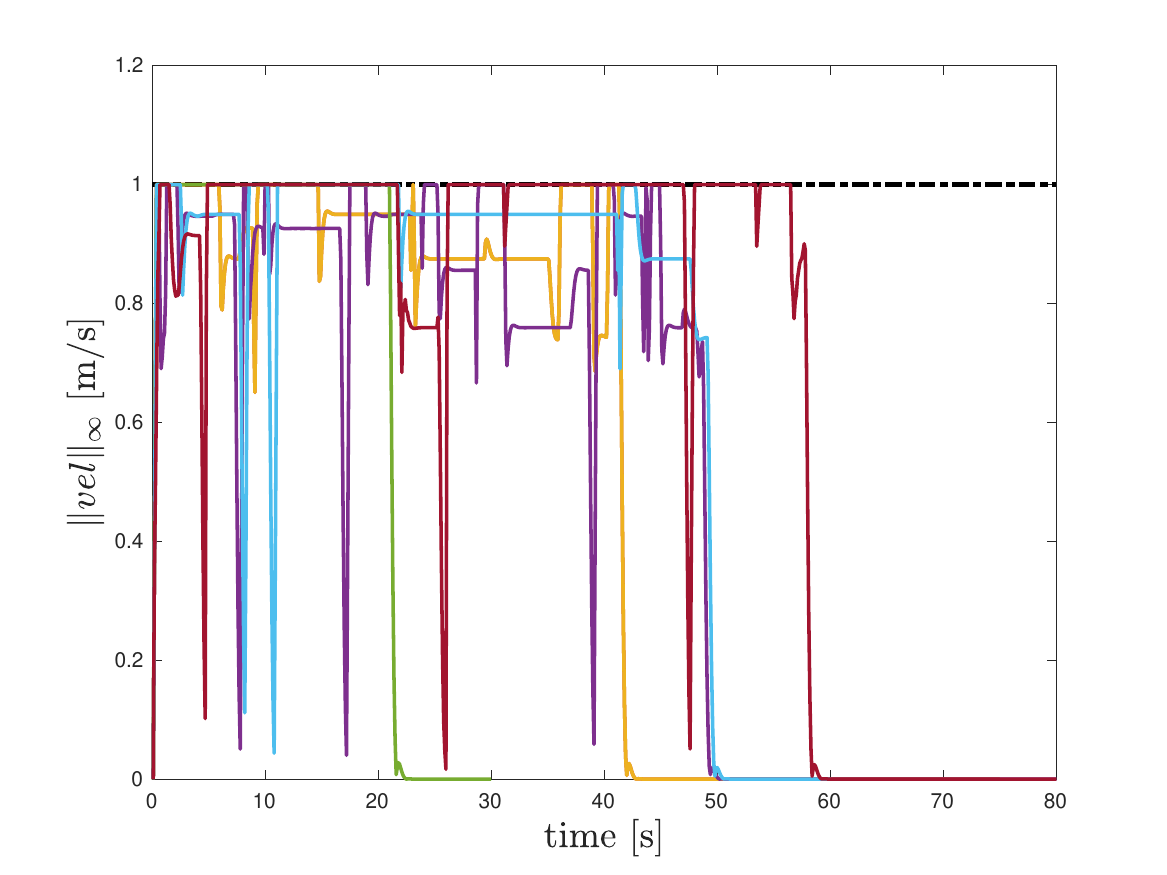}
    \caption{Typical histories of the pointwise-in-time infinity norm of velocity vector (solid lines) for the quadcopter tracking task. The bound on the infinity norm ($\|vel\|_\infty$) is represented by a dashed line.}
    \label{fig:drone_vels}
\end{figure}

\subsection{Computation times}\label{sec:compTimes}
As seen in Section \ref{sec:droneInCity}, the use of a reference governor leads to a low online computation burden. The goal of the present section is to investigate the offline computational load and, in particular, we examine how the the approach proposed in Section \ref{sec:hyperrectangles} compares to the process of computing the MOAS for each set (baseline approach). We do so for different lengths of the connected sequence and complexity of the adjacency matrix. Our results suggest that the approach proposed in Section \ref{sec:hyperrectangles} can substantially reduce the offline computational overhead\footnote{ Note that that approach was used both in Section \ref{sec:CWH} and Section \ref{sec:droneInCity}.}. 
For our study, we extend the 2D-MSD presented in Section \ref{sec:2DMSP} by considering an additional second order system with natural frequency and damping ratio given by $\omega_{{\tt n}3} = 5$ and $\zeta_3 = 0.02$, respectively. For this system, we impose a unit magnitude bound on the velocity components. Note that the resulting system is close to being marginally stable which led to a more complex MOAS (for the obstacles presented hereunder) than the one obtained when considering either the satellite or the drone dynamics.

In order to evaluate the offline computational overhead we  consider $N^{\rm obs}$ obstacles representing non-overlapping hyperrectangles with sides of length larger than 0.4 meters  and contained inside a hypercube centered at the origin with sides of length 5 meters. The feasible region is constructed using a naive algorithm that scans the three dimensional space along the different directions and generates a collection of connected hyperrectangles. Figure \ref{fig:exampleMultiObstcales} shows an example of the three dimensional space for $N^{\rm obs} = 30$.
\begin{figure}[b]
  \centering
 \includegraphics[width=1\linewidth]{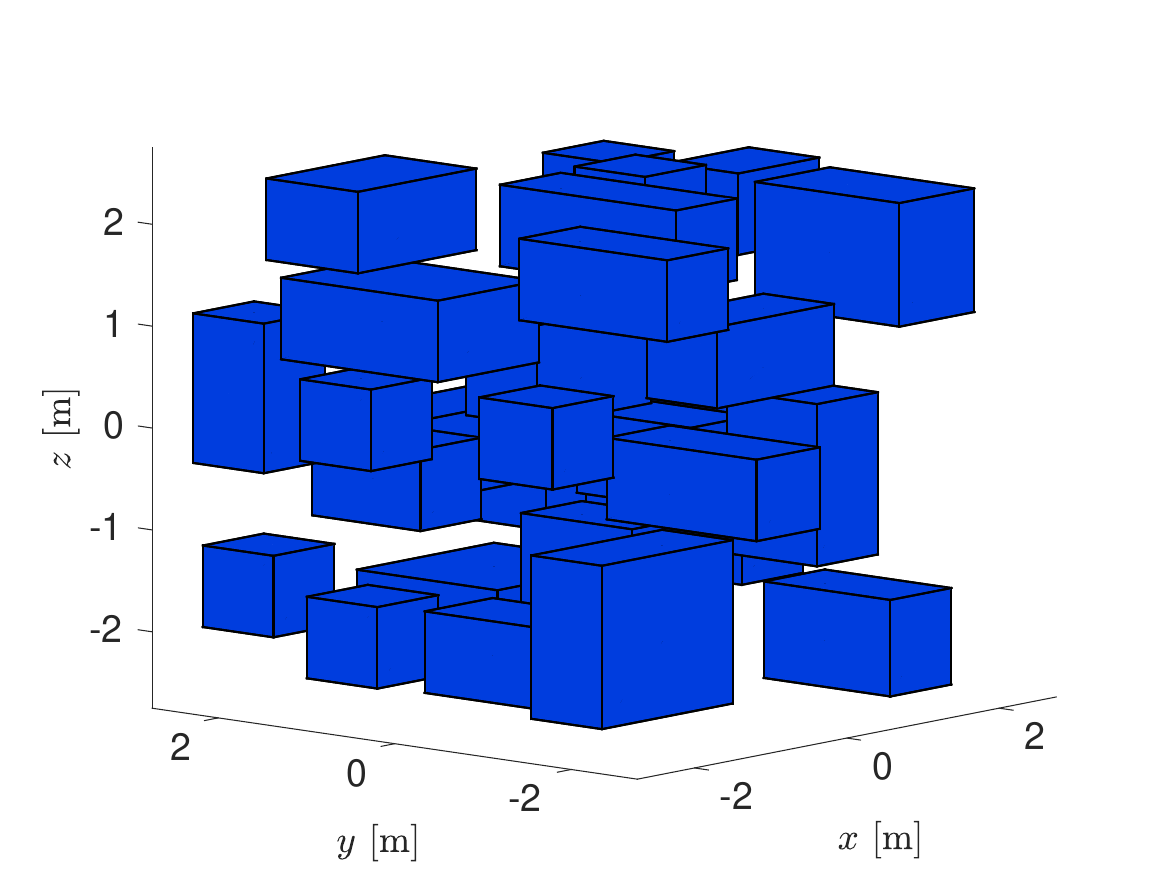}
    \caption{Randomly generated obstacles in three dimensional space for the experiment of Section \ref{sec:compTimes}.}
    \label{fig:exampleMultiObstcales}
\end{figure}

Data is collected for different values of $N^{\rm obs}$ and using 100 instances of the feasible space for each value.  Table \ref{tab:times} reports the mean and maximum values for different metrics: The volumetric ratio of occupied space inside the cube of length 5 m; the total number of elements in the connected sequence; the total number of sets for which the MOAS needs to be generated ($n_{\mathcal Y}/2$ + half the sum of all entries of the adjacency matrix, $F$); the time needed to generate the collection of polytopes from the obstacle information; the time needed to generate the adjacency matrix and weak extensions; the time needed to generate the family of MOAS; the total time (sum of the previous three metrics).

From Table \ref{tab:times}, we first note that as the number of obstacles increases the number of elements in the sequence ($n_{\mathcal Y}$) increases exponentially. Similarly, the total number of sets for which the MOAS needs to  be constructed (weak extensions + elements in the sequence) also increases exponentially. In contrast, the volumetric ratio increases linearly with the number of obstacles. While we reported the volumetric ratio, this metric is directly influenced by the imposed minimum obstacle length in every dimension. Similar experiments with larger values of the minimum obstacle length resulted in comparable computation times but increased volumetric ratios. Unsurprisingly, the method is agnostic to the volumetric ratio as the number of algebraic operations to be performed does not change.
\\
The time to generate the connected collection increased the most with the number of obstacles. This is understandable given that we are using a simple in-house algorithm. The time to generate the MOAS as well as the time needed to generate the weak extensions did not grow as fast, having a mean value of the order of $10$ and $100$ ms when considering 50 obstacles, respectively. When looking at the total time needed to generate the safe sets based on the obstacle information we see that on average handling a scene of 10 obstacles took 21 ms, suggesting that online safe set computation could be achievable when considering the method of Section \ref{sec:hyperrectangles}. 

In order to better quantify the improvement obtained with the proposed approach, the box plot in Figure \ref{fig:compTime} compares the time needed to scale the base MOAS to all sets using the approach we propose (Section \ref{sec:hyperrectangles}) to the time needed when computing all the MOAS individually. The latter is estimated by multiplying the total number of sets for which the MOAS needs to be computed by the time needed to compute the MOAS considering all constraints simultaneously and assuming a unit hypercube for the position constraints (10.87 s). In all cases, a difference of five orders of magnitudes was observed. Note that for all times related to the method of Section \ref{sec:hyperrectangles} we did not include the time needed to compute the MOAS for the base sets (3.26 s). This is reasonable as these are computed without any knowledge of the environment.
\begin{table}[]
\caption{Different metrics relating to the offline overhead as a function of the number of obstacles }
\label{tab:times}
\begin{tabular}{ll|llll}
 &$N^{\rm obs}$& 5 & 10 & 30  &50  \\\hline \arrayrulecolor{lightgray}
 \multirow{2}{*}{$n_{\mathcal Y}$}& mean & 24 & 52  & 207 & 397\\
 &max & 31 & 65  & 238 & 463  \\ \hline
 \multirow{2}{*}{ Total number of sets}& mean & 83 & 202  & 850 & 1657\\
 &max & 120 & 270  & 993 & 1939  \\ \hline
 \multirow{2}{*}{ $V^{\rm obstacles}/V^{\rm total}$ (\%)}& mean & 4 & 8  & 18 & 25\\
 &max & 8 & 15  & 27 & 36  \\\hline
 \multirow{2}{*}{Time to generate $\{\mathcal Y_i\}$ (ms)} & mean & 1    &4 &    450   & 5584\\
 & max & 18 &    9&    875 & 12,474\\\hline
 \multirow{2}{*}{Time to generate $\{\mathcal Y_{i,j}\}$ (ms)}& mean & 6 & 14 & 88  &  242  \\
 & max & 33 & 20 &  107 & 300 \\\hline
 \multirow{2}{*}{ Time to generate $\{\mathcal O_{i,j}\}$ (ms)}& mean & 2 & 3 &  20 &  55\\
 & max & 25 & 4 &  26 & 74 \\\hline
 \multirow{2}{*}{ \textbf{Total time }(ms)}& mean & 9 & 21 &  558 &  5882\\
 & max & 75 & 31 &  998 & 12,814  
\end{tabular}
\end{table}

\begin{figure}[h]
  \centering
 \includegraphics[width=1\linewidth]{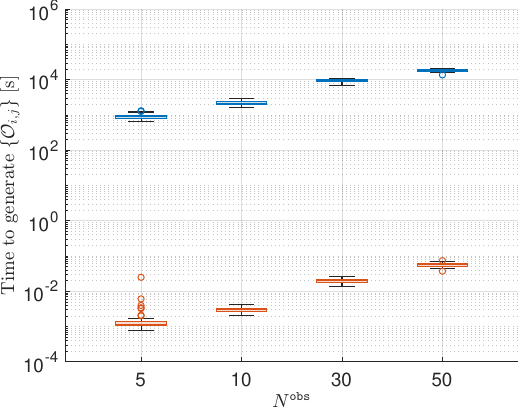}
    \caption{Box plot of the  time needed to generate all the MOAS for different numbers of obstacles. In blue, times computing individual MOAS. In orange, times when using the approach of Section \ref{sec:hyperrectangles}. Circular markers are outliers, whiskers show min-max values and boxes show median and lower/upper quartiles.}
    \label{fig:compTime}
\end{figure}
\section{Conclusion}
In this article we studied the problem of reference tracking for discrete-time linear systems subject to pointwise-in-time constraints that can be represented by a union of polytopes. We introduced the notion of simply/strictly connected collection of polytopes, which is a collection of polytopes with elements overlapping with one another on, at most, a single facet. In this setting, we developed a systematic method to generate a safe set with respect to the non-convex constraint set that enjoys the forward invariance property. We then showed how a reference governor scheme can be constructed, enabling safe tracking of any strictly feasible reference command. Theoretical guarantees for the supervisory scheme were derived, ensuring safety as well as finite-time convergence of the applied reference command to strictly admissible constant setpoints. Finally, conditions under which the safe sets for complex constraint sets can be computed efficiently offline were given.
Extensive numerical simulations demonstrated the applicability of the method in different fields. Offline and online computational times were reported showing promise for on-the-fly computation of the safe set. Based on these observations, future work will focus on developing dedicated methods to generate a connected collection of polytopes based on obstacle data as well as investigating the possibility for online implementation and handling of nonlinear systems. 

\appendix
The appendix provides an alternative proof for Corollary \ref{cor:extConvexity} (convexity of the extension) which does not rely on the knowledge that the extension is a polytope.

\begin{prop}
    \label{prop:extConvexity_alt}
     Let ${\mathcal Y_1,\;\mathcal Y_2}$ form a strictly connected collection of polytopes. Then, the extension of $\mathcal Y_1$ towards $\mathcal Y_2$ is convex.
\end{prop}
\begin{proof}
First, we note that if $\mathcal Y^{\tt res}_{2,1} = \emptyset$ then $\mathcal Y^{\tt e}_{1,2} = \mathcal Y_1$, which is convex. For the rest of this proof we focus on the case $\mathcal Y^{\tt res}_{2,1} \neq \emptyset$.
As $\mathcal Y^{\tt o}_{1,2}$ is a polytope it is convex, it then holds that $\mathcal Y^{\tt res}_{2,1}$ is also convex. To show $\mathcal Y^{\tt e}_{1,2}$ is convex, pick any $y_1\in\mathcal Y_1$, $y_2\in\mathcal Y^{\tt res}_{2,1}$ we then show that for all $\kappa\in\mathbb R_{[0,1]}$, $y_3(\kappa)\in \mathcal Y^{\tt e}_{1,2}$, where $y_3(\kappa) \triangleq \kappa y_2 + (1-\kappa)y_1$. By definition, we have that $A^{\mathcal Y_1}_{(G_{1,2})} y_1\leq b^{\mathcal Y_1}_{(G_{1,2})}$ and $A^{\mathcal Y_1}_{(G_{1,2})} y_2\geq b^{\mathcal Y_1}_{(G_{1,2})}$. Then, there exists $\bar\kappa\in\mathbb R_{[0,1]}$ such that $A^{\mathcal Y_1}_{(G_{1,2})} y_3(\bar\kappa) = b^{\mathcal Y_1}_{(G_{1,2})}$.
It then holds that $y_3(\bar \kappa)\in\mathcal Y_1$. Indeed, $y_1,y_2\in\mathcal Y^{\tt o}_{1,2}$ therefore $y_3(\kappa)\in\mathcal Y^{\tt o}_{1,2}$ for all $\kappa\in\mathbb R_{[0,1]}$, and thus $A^{\mathcal Y_1}_{(-G_{1,2})}y_3(\kappa)\leq b^{\mathcal Y_1}_{(-G_{1,2})}$. Then, by definition of $\bar \kappa$ we get $y_3(\bar \kappa)\in\mathcal Y_1$. 
Moreover, we conclude that $y_3(\bar \kappa)\in\mathcal Y_2\cap\mathcal Y^{\tt o }_{1,2}$ as  $y_3(\bar \kappa) \in \mathcal Y_1\cap{\tt H}(A^{\mathcal Y_1}_{(G_{1,2})},b^{\mathcal Y_1}_{(G_{1,2})}) =  \mathcal Y_2\cap{\tt H}(A^{\mathcal Y_1}_{(G_{1,2})},b^{\mathcal Y_1}_{(G_{1,2})})$, by definition of a strictly connected collection. Using convexity of $\mathcal Y_1$ and convexity of $\mathcal Y_2\cap \mathcal Y^{\tt o}_{1,2}$ we directly get that for all $\kappa\leq\bar\kappa$ $y_3(\kappa)\in\mathcal Y_1\subset \mathcal Y^{\tt e}_{1,2}$ and for all $\kappa\geq\bar\kappa$ $y_3(\kappa)\in\mathcal Y_2\cap\mathcal Y^{\tt o}_{1,2}\subset\mathcal Y^{\tt e}_{1,2}$, which concludes the proof.
\end{proof}

\end{document}